\documentclass[runningheads]{llncs}
\usepackage[T1]{fontenc}
\usepackage{amsmath,amssymb,ascmac,amsfonts,mathtools,bussproofs,mathabx,thmtools}
\usepackage{mathpartir} %
\usepackage{graphicx}
\usepackage{listings,jvlisting}
\usepackage{mdframed}
\usepackage{color}
\usepackage[dvipsnames]{xcolor}
\usepackage{relsize}
\usepackage{stmaryrd}
\usepackage[shortlabels]{enumitem}
\usepackage{titletoc}
\usepackage{hyperref}
\hypersetup{
  colorlinks=true,
  linkcolor=RoyalBlue,
  citecolor=Green,
  urlcolor=RoyalBlue
}
\usepackage[nameinlink]{cleveref}
\usepackage{url}
\usepackage{multirow}
\usepackage[normalem]{ulem} %

\usepackage[fullpaper]{myenv}
\begin{document}
\title{Borrowable Fractional Ownership Types for Verification}
\author{Takashi Nakayama\inst{1} \and
Yusuke Matsushita\inst{1} \and
Ken Sakayori \inst{1} \and
Ryosuke Sato \inst{1} \and
Naoki Kobayashi\inst{1}
}
\authorrunning{T. Nakayama et al.}
\institute{
  The University of Tokyo, Tokyo, Japan\\
  \email{\{takashi-nakayama, yskm24t, sakayori, rsato, koba\}@is.s.u-tokyo.ac.jp}
}
\maketitle              %
\begin{abstract}
Automated verification of functional correctness of imperative programs with references (a.k.a. pointers) is challenging because of reference aliasing.
Ownership types have recently been applied to address this issue, but the existing approaches were limited in that they are effective only for a class of programs whose reference usage follows a certain style.
To relax the limitation, we combine the approaches of \consort{} (based on fractional ownership) and \rusthorn{} (based on borrowable ownership), two recent approaches to automated program verification based on ownership types, and propose the notion of \emph{borrowable fractional ownership types}.
We formalize a new type system based on the borrowable fractional ownership types and show how we can use it to automatically reduce the program verification problem for imperative programs with references to that for functional programs without references.
We also show the soundness of our type system and the translation, and conduct experiments to confirm the effectiveness of our approach.
\keywords{Automated program verification \and Ownership types \and Imperative programs}
\end{abstract}

\section{Introduction} \label{section:intro}
Various notions of ownership types have been used to prevent race conditions in concurrent programs and also to enable strong updates of knowledge in sequential programs by controlling references (or pointers) aliases~\cite{ref_boyapati_2002,ref_calcagno_2007,ref_sergey_2012,ref_bao_2021,ref_simuliris,ref_crichton_2022}.
Among others, Boyland~\cite{ref_boyland_03} introduced fractional ownership (a.k.a. fractional permissions), where a fractional value in \([0,1]\) is associated with each reference and the full ownership (of \(1\)) allows write access whereas a non-zero ownership allows read access.
The Rust programming language~\cite{ref_rust_site,ref_rust_paper} incorporates ownership with the \emph{borrow} mechanism (which we call \emph{borrowable ownership}), where ownership for a mutable reference can be \emph{borrowed} during a specific lifetime, and the borrowed reference can be used for write access during the lifetime.

Both notions of ownership (i.e., fractional ownership and borrowable ownership) have recently been used for fully automated verification of the functional correctness of imperative programs~\cite{ref_consort,ref_rusthorn,ref_creusot,ref_aeneas,ref_prusti,ref_flux}.
Among them, \consort{}~\cite{ref_consort} used fractional ownership to enable strong updates of refinement types for automated verification of the absence of assertion failures, whereas \rusthorn{}~\cite{ref_rusthorn} used borrowable ownership to reduce assertion verification to CHC solving.
Although both approaches have been shown effective, they also have unavoidable limitations: they are inadequate when a given program does not follow a specific pattern of reference usage, i.e., when the type system conservatively rejects the program.
Therefore, relaxing this limitation by extending the expressive power of ownership type systems is essential for automated verification of imperative programs leveraging ownership types.

\newcommand{\refty}[2]{{#1}\,\mathbf{ref}^{\,{#2}}}
\newcommand{\csirefty}[2]{\refty{\{\nu:\tint \mid {#1}\}}{{#2}}}
\newcommand{\rsirefty}[1]{\refty{\tint}{{#1}}}
\newcommand{\Alias}[1]{\texttt{alias}}
\newcommand{\TRUE}[1]{\texttt{true}}

In this paper, we propose the notion of \emph{borrowable fractional ownership}, which combines both approaches to extend the expressive power of ownership types and develop a new method to verify the functional correctness of programs.
Hereafter, we first review the ideas of \consort{} and \rusthorn{}, showing how their ownership types work for automated verification, and then explain our new approach.

In \consort{}, each reference is given a type of the form \(\refty{\tau}{r}\), where \(r\), called (fractional) ownership, ranges over \([0,1]\), and \(\tau\) is a refinement type that describes the content of the reference.
Additionally, a reference type should satisfy the well-formedness that no refinement information (that is, only \(\TRUE{}\)) is allowed for zero-ownership references.
\cref{code:consort-demo} shows an example program of \consort{} and how it is typed.
The type of \(x\) on the first line shows that \(x\) is a reference with full ownership, pointing to a cell holding the value \(0\).
On the second line, the aliasing reference \(y\) is created, and all the ownership has been transferred to \(y\), which deprives \(x\) of refinement information in its type.
On the third line, the type of \(y\) can be strongly updated, as \(y\) has the full ownership.
The \Alias{} command on the fourth line tells the type system that \(x\) and \(y\) are aliases; based on that, the information on \(y\) is propagated to \(x\), and the ownership of \(y\) and its refinement information is redistributed between \(x\) and \(y\).
The fifth line type-checks, as the types of \(x\) and \(y\) say that they both store \(1\).
In this manner, \consort{} uses fractional ownership to allow strong updates of refinement types and share information among references based on \Alias{} annotations (which may automatically be inserted by a separate pointer analysis).

\begin{figure}
\begin{lstlisting}[style=mystyle]
let x = mkref 0 in // $\color{Blue} x:\csirefty{\nu=0}{1}$
let y = x in       // $\color{Blue} x:\csirefty{\TRUE{}}{0},\ y:\csirefty{\nu=\textcolor{red}{0}}{1}$
y := 1;            // $\color{Blue} x:\csirefty{\TRUE{}}{0},\ y:\csirefty{\nu=\textcolor{red}{1}}{1}$
alias(x = y);      // $\color{Blue} x:\csirefty{\nu=1}{0.5},\ y:\csirefty{\nu=1}{0.5}$
assert( *x + *y = 2 ) // $\color{Blue}{\texttt{ok}}$
\end{lstlisting}
    \caption{A program example of \consort{}}
    \label{code:consort-demo}
\end{figure}

\begin{figure}
\begin{lstlisting}[style=mystyle]
let x = mkref 0 in // $\color{Blue} x:\rsirefty{\alpha}$
let y = x in       // $\color{Blue} y:\rsirefty{\beta}$; $\color{Blue} x$ is invalidated during $\color{Blue} \beta$
y := 1;
endlft $\beta$;          // dispose $\color{Blue} y$; $\color{Blue} x:\rsirefty{\alpha}$
assert( *x = 1 )
\end{lstlisting}
    \caption{Typing in \rusthorn{}}
    \label{code:rusthorn}
\end{figure}

On the other hand, the type system of \rusthorn{} (which inherits that of Rust \footnote{For formalization, \rusthorn{} inherits the type system of \(\lamrust\), a core calculus for Rust by Jung et al.~\cite{ref_rustbelt}.}) expresses the reference type as of the form \(\tau\,\rawtref{\alpha}{}\), where \(\alpha\) is a \emph{lifetime}, a symbol representing an abstract time period during which the reference is valid.
In addition, the borrowed references are invalidated during the corresponding lifetime in Rust.
\cref{code:rusthorn} shows a variation of the program in \cref{code:consort-demo} and how it is typed in the type system of \rusthorn{}.
On the first line, a reference \(x\) with lifetime \(\alpha\) is created.
On the second line, \(y\) is created as an alias, and the (full) ownership of \(x\) is \emph{borrowed} to \(y\) during the lifetime \(\beta\), which invalidates \(x\).
On the third line, \(y\) can be safely updated as it has temporally borrowed ownership.
On the fourth line, the lifetime \(\beta\) ends, and the borrowed ownership is returned to \(x\), and thus \(x\) can be safely accessed on the fifth line.
In reality, \texttt{endlft} commands are automatically inserted by the Rust compiler.

To verify the assertions, \rusthorn{} uses the technique of \emph{prophecy}~\cite{ref_abadi_1991,ref_calcagno_2007,ref_jung_2019} to reduce the problem of verification to that of CHC solving.
Here, for the sake of clarity, we instead reduce them to the verification problem for functional programs.
In the reduction, a mutably borrowing reference \(y\) is modeled as a pair of \(\tuple{v, \proph{y}}\), where the first component \(v\) is the current value of the reference \(y\) and the second component \(\proph{y}\), called \(y\)'s \emph{prophecy}, expresses the future value at the end of \(y\)'s lifetime.
For verification, the prophecy is encoded by standard non-determinism.
\cref{code:rusthorn-red} shows how the translation works for the program in \cref{code:rusthorn}.
Notice that on the second line, the mutably borrowing reference \(y\) becomes a pair of the current value (\(x\)) and the prophecy value initialized non-deterministically (represented by \(\_\)), then \(x\) takes the prophecy value of \(y\) on the third line.
After the value of \(y\) is changed on the fourth line, the \texttt{assume} command on the fifth line, which is converted from \(\lftend \beta\) in the original program, \emph{resolves} the prophecy value of \(y\); in other words, the formally non-deterministic second value of \(y\) is determined by the current value (\(1\)) so that \(x\) recover its correct value and the last assertion succeeds.
Note that this program has no references and is easier to verify than the original one.

\begin{figure}
\begin{lstlisting}[style=mystyle]
let x = 0 in           // $\color{Blue} x = 0$
let y = (x, _) in      // $\color{Blue} x = 0, y = \tuple{0, \_}$
let x = snd y in       // $\color{Blue} x = {\color{Red}\_}, y = \tuple{0, \_}$
let y = (1, snd y) in  // $\color{Blue} x = \_, y = \tuple{{\color{Red}1}, \_}$
assume(fst y = snd y); // $\color{Blue} x = {\color{Red}1}, y = \tuple{1, {\color{Red}1}}$
assert(x = 1)  // ok
\end{lstlisting}
   \caption{A pure functional program equivalent to \cref{code:rusthorn}}
   \label{code:rusthorn-red}
\end{figure}

As mentioned, both approaches have their limitations: \consort{} relies on \Alias{} annotations and thus works effectively only for programs where exact alias information is available statically (\cref{code:demo-combination-intro} shows one such program), and \rusthorn{} relies on Rust's type system, hence is inapplicable to programs written in other imperative languages like C.

In this paper, to overcome the limitations and take the best of both approaches, we introduce \emph{borrowable fractional ownership types}.
The type of a reference is now of the form \(\tref{\alpha}{r}{B}\), where \(\alpha\) is the lifetime of the reference, \(r\) is the (fractional) ownership of the reference ranging over \([0, 1]\), and \(B\) is the borrowing information of the reference.
A borrowing \(B\) is either a pair \((\beta, s)\) or \(\emptyset\), where \(\beta\) is the lifetime for which the reference lends the ownership and \(s\) is the amount of (fractional) ownership others borrow from the reference.
Our type system extends those of both \consort{} and Rust ; for the sake of simplicity, however, we consider only integer references, excluding nested reference types like \(\tint\rawtref{\alpha,r}{B_1}\rawtref{\beta,s}{B_2}\). (We briefly discuss this in \cref{section:conclusion}.)
The reference types of \consort{} can be considered a particular case where \(B\) is always \(\emptyset\), and all the lifetimes are identical.\footnote{Unlike \consort{}, we do not have refinement predicates because we follow the translation-based approach of \rusthorn{} for verification.}
And those of Rust can be considered a particular case where the ownership \(r\) is always \(0\) or \(1\). \footnote{This is not the case for shared (a.k.a. immutable) references in Rust, but our type system also subsumes these references.}
\cref{code:demo-combination-intro} shows an example program that can be handled by neither \consort{} nor \rusthorn{} but by our new type system. (We will show how the program is typed in our type system in \cref{section:type-system}.)
Together with the type system, we develop a verification method to prove the absence of assertion failures by a type-directed translation from imperative programs with mutable references into stateless functional programs without references, in a similar way to \rusthorn{}.

\begin{figure}
\begin{lstlisting}[style=mystyle]
minmax(x, y) { if *x < *y then (x, y) else (y, x) }
rand_choose(x, y) { if _ then x else y }
let x = mkref _ in  let y = mkref _ in
let (p, q) = minmax(x, y) in
let z = rand_choose(x, y) in
assert( *p <= *z && *z <= *q );
z := 1
\end{lstlisting}
    \caption{Example program that can be handled by neither \consort{} nor \rusthorn{} but by our method}
    \label{code:demo-combination-intro}
\end{figure}

Our main contributions are:
\begin{enumerate}
\item Proposal of the new notion of borrowable fractional ownership.
\item Formalization of a type system with borrowable fractional ownership and proof of the preservation of types and ownership invariants.
\item Type-directed translation from imperative programs to functional programs without references, and proof of its soundness: if the target program has no assertion failures, neither does the source program.
\end{enumerate}

The rest of this paper is organized as follows.
\cref{section:source-lang} introduces the source language.
\cref{section:type-system} formalizes the type system with borrowable fractional ownership.
\cref{section:translation} shows the reduction to stateless functional programs.
\cref{section:experiments} reports preliminary experimental results.
\cref{section:related} discusses related work and \cref{section:conclusion} concludes the paper and discusses future work.

\section{Source Language} \label{section:source-lang}
This section introduces a simple imperative language with mutable references and recursive functions with lifetime polymorphism,
which serves as the target of our type-based verification method.

\subsection{Syntax}

The syntax of the language is as follows:
\begin{gather*}
    o \mbox{ (arithmetic expressions)} \Coloneqq n \mid x \mid o_1 \op o_2
\\[.2em]
    \begin{aligned}
        e \mbox{ (expressions)} \Coloneqq\ & x \mid \letin{x}{o} e \mid \letin{x}{y} e \\
        \mid\ & \letin{x}{\mkref y} e \mid \letin{x}{\star y} e \mid x := y; e \\
        \mid\ & \ifz x \then e_1 \els e_2 \mid \letin{x}{\fcall{f}{\vv{\alpha}}{y_1, \dots, y_n}} e \\
        \mid\ & \alias{x = y}; e \mid \newlftin{\alpha} e \mid \lftend \alpha; e \mid \fail
    \end{aligned}
\\[.2em]
    d \mbox{ (function definitions)} \ \Coloneqq\ f \mapsto \tuple{\vv{\alpha}}(x_1, \dots, x_n) e
\\[.2em]
    P \mbox{ (programs)} \ \Coloneqq\ \langle \lbrace d_1, \dots, d_n \rbrace, e \rangle
\end{gather*}
The meta-variables \(x, y, \dots \) range over the set \(\setvar\) of program variables,
\(f\) ranges over the set of function names, and
\(\alpha, \beta \in \setlft\) range over a set of symbols called \emph{lifetime variables}.

Program variables are introduced by function parameters or let bindings, and the lifetime variables are bound by \(\newlft\).
We assume that all program and lifetime variables are alpha-renamed as necessary,
so that each variable is unique for each binder.

For simplicity, an expression \(e\) is restricted to a de-nested, simplified form.
An arithmetic expression \(o\) consists of integer constants, integer variables, and arithmetic binary operations (\(+\), \(-\), etc.) denoted by \(\op\).
The expression \(\letin{x}{y} e\) creates a new alias \(x\) of \(y\) and evaluates \(e\).
We have primitives  \(\mkref y\) for creating a new reference,  \(\star y\) for reading from a reference, and \(x := y\) for updating a reference.
The conditional branch \(\ifz x \then e_1 \els e_2\) evaluates \(e_1\) if the value of \(z\) is \(0\), and evaluates \(e_2\) otherwise.
A function call \(\fcall{f}{\vv{\alpha}}{y_1, \dots, y_n}\) takes lifetime arguments \(\vv{\alpha}\) along with program arguments \(y_1, \dots, y_n\).
The expression \(\fail\) aborts program execution with an error.

An \emph{alias assumption} \(\alias{x = y}\) has been inherited from \consort{}\cite{ref_consort}; it assumes that \(x\) and \(y\) are aliasing references.
This information is exploited by our type system described in \cref{section:type-system}.
Moreover, we inherit \(\newlft\ \alpha\) and \(\lftend \alpha\), which respectively introduces and eliminates a lifetime \(\alpha\), from \(\lambda_{\text{Rust}}\)~\cite{ref_rustbelt}.
They play essential roles in our type system but do not affect program execution.

A function declaration \(d\) consists of the function name \(f\), the parameter variables \(x_1, \dots, x_n\), and the function body \(e\).
We allow mutual recursion between functions.
A program \(P\) is a pair \(\langle D, e \rangle\) where \(D = \{d_1, \dots, d_n\}\) is a set of function definitions and \(e\) is the main expression.

Henceforth, we write programs in an abbreviated, sugared style to avoid unnecessary complications.
For example, the program in \cref{code:consort-demo} is a valid abbreviated program in our source language.
Here, the expression \(\assert{e_1 = e_2}\) is syntactic sugar of \((\letin{c}{e_1 - e_2} \ifz c \then 0 \els \fail)\).

\subsection{Operational Semantics}

We now introduce the operational semantics of our source language.
Let \(\addr\) be a countable set of heap addresses.
We define a set of runtime values \(\srcval\) as \(\mathbb{Z} \cup \addr\).

A configuration (runtime state) of this language is a quadruple \(\configsrc{H}{R}{\vv{F}}{e}\), consisting of a heap, a register, a call stack, and a currently reducing expression.
A heap \(H\) is a partial function from \(\addr\) to \(\srcval\).
A register \(R\) is a partial mapping from \(\setvar\) to \(\srcval\).
A call stack \(\vv{F}\) is a sequence of return contexts of the form \((\letin{x}{[]} e)\).
A program \(\langle D, e \rangle\) is executed by repeatedly stepping from the initial configuration \(\configsrc{\emptyset}{\emptyset}{\cdot}{e}\) according to the one-step reduction relation \(\reduce{D}\), which is defined by the rules in \cref{fig:source-opsem}.
\begin{figure}
    \begin{center}
    \AxiomC{}
    \RightLabel{\smaller (\textsc{Rs-Var})}
    \UnaryInfC{\(\configsrc{H}{R}{(\letin{x'}{[]} e) : \vv{F}}{x} \reduce{D} \configsrc{H}{R}{\vv{F}}{[x/x']e}\)}
    \DisplayProof
    \end{center}
    \begin{center}
        \AxiomC{\(\llbracket o \rrbracket_R = n\)}
        \AxiomC{\(x' \notin \dom(R)\)}
    \RightLabel{\smaller (\textsc{Rs-Arith})}
    \BinaryInfC{\(\configsrc{H}{R}{\vv{F}}{\letin{x}{o}e} \reduce{D} \configsrc{H}{R\{x' \mapsto n\}}{\vv{F}}{[x'/x]e}\)}
    \DisplayProof
    \end{center}
    \begin{center}
    \AxiomC{\(x' \notin \dom(R)\)}
    \RightLabel{\smaller (\textsc{Rs-Let})}
    \UnaryInfC{\(\configsrc{H}{R}{\vv{F}}{\letin{x}{y}e} \reduce{D} \configsrc{H}{R\{x' \mapsto R(y)\}}{\vv{F}}{[x'/x]e}\)}
    \DisplayProof
    \end{center}
    \begin{center}
        \AxiomC{\(f \mapsto \tuple{\vv{\alpha}}(x_1, \dots, x_n) e \in D\)}
        \AxiomC{\(x' \notin \dom(R)\)}
    \RightLabel{\smaller (\textsc{Rs-Call})}
    \BinaryInfC{\(\begin{array}{l}
        \configsrc{H}{R}{\vv{F}}{\letin{x}{f\tuple{\vv{\alpha}}(y_1, \dots, y_n)}e'} \\
        \quad \reduce{D} \configsrc{H}{R}{(\letin{x'}{[]}[x'/x]e'):\vv{F}}{[y_1/x_1]\cdots[y_n/x_n]e}
    \end{array}\)}
    \DisplayProof
    \end{center}
    \begin{center}
        \AxiomC{\(a \notin \text{dom}(H)\)}
        \AxiomC{\(x' \notin \dom(R)\)}
    \RightLabel{\smaller (\textsc{Rs-MkRef})}
    \BinaryInfC{\(\begin{array}{l}
        \configsrc{H}{R}{\vv{F}}{\letin{x}{\mkref y}e} \\
        \quad \reduce{D} \configsrc{H\{a \mapsto R(y)\}}{R\{x' \mapsto a\}}{\vv{F}}{[x'/x]e}
    \end{array}\)}
    \DisplayProof
    \end{center}
    \begin{center}
        \AxiomC{\(R(y) = a\)}
        \AxiomC{\(H(a) = v\)}
        \AxiomC{\(x' \notin \dom(R)\)}
    \RightLabel{\smaller (\textsc{Rs-Deref})}
    \TrinaryInfC{\(\configsrc{H}{R}{\vv{F}}{\letin{x}{\star y}e} \reduce{D} \configsrc{H}{R\{x' \mapsto v\}}{\vv{F}}{[x'/x]e}\)}
    \DisplayProof
    \end{center}
    \begin{center}
        \AxiomC{\(R(x) = a\)}
        \AxiomC{\(a \in \text{dom}(H)\)}
    \RightLabel{\smaller (\textsc{Rs-Assign})}
    \BinaryInfC{\(\configsrc{H}{R}{\vv{F}}{x := y; e} \reduce{D} \configsrc{H\{a \mapsto R(y)\}}{R}{\vv{F}}{e}\)}
    \DisplayProof
    \end{center}
    \begin{center}
    \AxiomC{\(R(x) = 0\)}
    \RightLabel{\smaller (\textsc{Rs-IfTrue})}
    \UnaryInfC{\(\configsrc{H}{R}{\vv{F}}{\ifz x \then e_1 \els e_2} \reduce{D} \configsrc{H}{R}{\vv{F}}{e_1}\)}
    \DisplayProof
    \end{center}
    \begin{center}
    \AxiomC{\(R(x) \neq 0\)}
    \RightLabel{\smaller (\textsc{Rs-IfFalse})}
    \UnaryInfC{\(\configsrc{H}{R}{\vv{F}}{\ifz x \then e_1 \els e_2} \reduce{D} \configsrc{H}{R}{\vv{F}}{e_2}\)}
    \DisplayProof
    \end{center}
    \begin{center}
    \AxiomC{\(R(x) = R(y)\)}
    \RightLabel{\smaller (\textsc{Rs-Alias})}
    \UnaryInfC{\(\configsrc{H}{R}{\vv{F}}{\alias{x = y}; e} \reduce{D} \configsrc{H}{R}{\vv{F}}{e}\)}
    \DisplayProof
    \end{center}
    \begin{center}
    \AxiomC{\(R(x) \neq R(y)\)}
    \RightLabel{\smaller (\textsc{Rs-AliasFail})}
    \UnaryInfC{\(\configsrc{H}{R}{\vv{F}}{\alias{x = y}; e} \reduce{D} \textbf{AliasFail}\)}
    \DisplayProof
    \end{center}
    \begin{center}
    \AxiomC{}
    \RightLabel{\smaller (\textsc{Rs-Newlft})}
    \UnaryInfC{\(\configsrc{H}{R}{\vv{F}}{\newlftin{\alpha} e} \reduce{D} \configsrc{H}{R}{\vv{F}}{e}\)}
    \DisplayProof
    \end{center}
    \begin{center}
    \AxiomC{}
    \RightLabel{\smaller (\textsc{Rs-Endlft})}
    \UnaryInfC{\(\configsrc{H}{R}{\vv{F}}{\lftend \alpha; e} \reduce{D} \configsrc{H}{R}{\vv{F}}{e}\)}
    \DisplayProof
    \end{center}
    \begin{center}
    \AxiomC{}
    \RightLabel{\smaller (\textsc{Rs-Fail})}
    \UnaryInfC{\(\configsrc{H}{R}{\vv{F}}{\fail} \reduce{D} \textbf{Fail}\)}
    \DisplayProof
    \end{center}
    \caption{Operational semantics of the source language}
    \label{fig:source-opsem}
\end{figure}
We write \(H\{a \mapsto v\}\) for the heap that maps the address \(a\) to \(v\) and behaves as \(H\) for the other addresses.
We also use a similar notation \(R\{x \mapsto v\}\) for registers.
In \drvrule{Rs-Arith}, we write \(\llbracket o \rrbracket_R\) for the computed integer value of \(o\), where variables are interpreted according to the register \(R\).
By \drvrule{Rs-Fail}, we fall into the `hard' failure state \textbf{Fail} out of \(\fail\).
Also, by \drvrule{Rs-AliasFail}, in case the alias assumption \(\alias{x = y}\) is not satisfied, we fall into \textbf{AliasFail}, a `soft' failure state distinct from \textbf{Fail}.

The goal of our verification method is to check that a given program does not reach \textbf{Fail}.
To that end, we introduce a type system based on the new notion of borrowable fractional ownership types in the next section and use it to reduce the verification problem to that for functional programs in \cref{section:translation}.

\section{Type System} \label{section:type-system}
In this section, we introduce our borrowable fractional ownership type system for the source language.
This type system will be used for the translation described later in \cref{section:translation} to verify programs.

\subsection{Syntax of Types}

The syntax of types is given as follows:
\begin{align*}
    \tau \mbox{ (types)} \Coloneqq\ & \tint \mid \tref{\alpha}{r}{B} \\
    r \mbox{ (ownership)}\ \in\ \  & [0, 1] \\
    B\mbox{ (borrowings)} \Coloneqq\ & \emptyset \mid (\alpha, r) \\
    \sigma\mbox{ (function types)} \Coloneqq\ & \forall \vv{\alpha}:\mathcal{L}.\ \langle \tau_1, \dots, \tau_n \rangle \to \langle \tau_1', \dots, \tau_n' \mid \tau \rangle \\
    \mathcal{L}\mbox{ (lifetime environment)} \Coloneqq\ &\emptyset \mid \mathcal{L}, \alpha_1 \lftlt \alpha_2 \\
    \Gamma\mbox{ (type environment)} \Coloneqq\ &\emptyset \mid \Gamma, x:\tau
\end{align*}

Our type system has the integer type \(\tint\) and the reference type of integer \(\tref{\alpha}{r}{B}\), where \(\alpha, r\) indicates that the reference has the lifetime \(\alpha\) and the fractional ownership \(r\).
The \emph{borrowing} \(B\) specifies how much ownership is borrowed by variables during which lifetime.
If \(B = \emptyset\), no borrowing for other lifetime occurred.
Otherwise, \(B = (\beta, s)\), which indicates that the reference is lending ownership \(s\) to variables with lifetime \(\beta\).
For example, \(\tref{\alpha}{0.5}{\beta,0.5}\) is a reference with the lifetime \(\alpha\) that has the ownership of \(0.5\) and lends the ownership of \(0.5\) to other variables during the lifetime \(\beta\).

A \emph{lifetime environment} \(\mathcal{L}\) is a strict partial order on the set of valid lifetime variables, each element of which is denoted \(\alpha \lftlt \beta\).
We impose the constraint that when a reference type \(\tref{\alpha}{r}{\beta, s}\) exists, the lifetimes \(\alpha\) and \(\beta\) are ordered \(\beta \lftlt \alpha\) under the lifetime environment \(\mathcal{L}\) at the point. A \emph{type environment} \(\Gamma\) is a finite set of type bindings of the form \(x:\tau\).

A \emph{function type} \(\sigma\) takes the form \(\forall \vv{\alpha}:\mathcal{L}.\ \langle \tau_1, \dots, \tau_n \rangle \to \langle \tau_1', \dots, \tau_n' \mid \tau \rangle\).
This indicates that the \(i\)-th argument has the type \(\tau_i\) before the function call and changes its type to \(\tau_i'\) when the function returns.
In addition, the function type has a direct return type \(\tau\).
The function types are parameterized over lifetime variables \(\vv{\alpha}\) with an ordering \(\mathcal{L}\), and the lifetime variables appearing in \(\tau, \tau_i, \tau_i'\) must be included in \(\vv{\alpha}\).

\paragraph{Notations.} Hereafter, we will identify the borrowing of \((\alpha, 0)\) with \(\emptyset\) and denote \(\tref{\alpha}{r}{\emptyset}\) by \(\tref{\alpha}{r}{}\).
Given a reference type \(\tau = \tref{\alpha}{r}{B}\), we define \( \lft(\tau) \defrel \alpha \) and \( \own(\tau) \defrel r \); \( \lft(\tint) \) is undefined and we define \( \own(\tint) \defrel 0 \) for technical convenience.
For \( x : \tau \in \Gamma \), we  write \( \lft_{\Gamma}(x) \) and  \( \own_{\Gamma}(x) \) for  \( \lft(\tau)\) and \( \own(\tau) \), respectively.

\subsection{Typing Rules}
The typing rules for expressions are defined in \cref{fig:typing-standard,fig:typing-reference,fig:typing-ghost-inst}.
A type judgment for expressions is of the form \( \tyjudgement{\Theta}{\mathcal{L}}{\Gamma}{e:\tau}{\mathcal{L}'}{\Gamma'} \)
where \(\Theta\), called a \emph{function type environment}, is a map from function names, ranging over by \( f \), to function types.
The judgment indicates that an expression \(e\) is well-typed with a type \(\tau\) under the environments \(\Theta\), \(\mathcal{L}\), and \(\Gamma\), and further that the lifetime and type environments change to \(\mathcal{L}'\) and \(\Gamma'\) after evaluating \(e\).

We tacitly assume that every type judgment is \emph{well-formed}.
A type judgment \(\tyjudgement{\Theta}{\mathcal{L}}{\Gamma}{e:\tau}{\mathcal{L}'}{\Gamma'}\) is well-formed if \(\mathcal{L} \vdash_{\text{WF}} \Gamma\), \(\mathcal{L'} \vdash_{\text{WF}} \Gamma'\), and \(\mathcal{L} \vdash_{\text{WF}} \tau\), which mean the type \(\tau\) or every type used in \(\Gamma, \Gamma'\) should satisfy the order of the corresponding lifetime environment, \(\mathcal{L}\) or \(\mathcal{L}'\), and every reference type \(\tref{\alpha}{r}{\beta,s}\) should satisfy \(r + s \le 1\).
For example, if \(\beta \lftlt \alpha \in \mathcal{L}\), then a type \(\tref{\beta}{r}{\alpha,s}\) is invalid because the smaller lifetime \(\beta\) lends ownership to the larger lifetime \(\alpha\). (The necessity of the ordering is discussed in \cref{rem:order-of-lifetime}.)

\begin{figure}[t]
\begin{minipage}{0.54\linewidth}
    \AxiomC{}
    \RightLabel{\smaller (\textsc{T-Var})}
    \UnaryInfC{\(\tyjudgement{\Theta}{\mathcal{L}}{\Gamma + \Gamma' +  x \!:\!\tau}{x \!:\! \tau }{\mathcal{L}}{\Gamma'}\)}
    \DisplayProof
\end{minipage}
\begin{minipage}{0.45\linewidth}
\AxiomC{}
\RightLabel{\smaller (\textsc{T-Fail})}
\UnaryInfC{\(\tyjud{\Gamma}{\fail \!:\! \tau}{\Gamma'}\)}
\DisplayProof
\end{minipage}
\
\begin{center}
    \AxiomC{\(\tau = \tau_x+\tau_y\)}
    \AxiomC{\( x \notin \dom(\Gamma')\)}
    \noLine
    \BinaryInfC{\(\tyjud{\Gamma, x:\tau_x, y:\tau_y}{e : \rho}{\Gamma'}\)}
\RightLabel{\smaller (\textsc{T-Let})}
\UnaryInfC{\(\tyjud{\Gamma, y:\tau}{\letin{x}{y} e : \rho}{\Gamma'}\)}
\DisplayProof
\end{center}
\begin{center}
    \AxiomC{\(\Gamma \vdash o : \tint\)}
    \AxiomC{\(\tyjud{\Gamma, x : \tint}{e : \rho}{\Gamma'}\)}
    \AxiomC{\( x \notin \dom(\Gamma') \)}
\RightLabel{\smaller (\textsc{T-Arith})}
\TrinaryInfC{\(\tyjud{\Gamma}{\letin{x}{o} e : \rho}{\Gamma'}\)}
\DisplayProof
\end{center}
\begin{center}
    \AxiomC{\(\tyjud{\Gamma, x:\tint}{e_1 : \rho}{\Gamma'}\)}
    \AxiomC{\(\tyjud{\Gamma, x:\tint}{e_2 : \rho}{\Gamma'}\)}
\RightLabel{\smaller (\textsc{T-If})}
\BinaryInfC{\(\tyjud{\Gamma, x:\tint}{\ifz x \then e_1 \els e_2 : \rho}{\Gamma'}\)}
\DisplayProof
\end{center}
\begin{center}
    \AxiomC{\(\rho = [\vv{\beta}/\vv{\alpha}]\tau, \rho_i = [\vv{\beta}/\vv{\alpha}]\tau_i, \rho_i' = [\vv{\beta}/\vv{\alpha}]\tau_i'\)}
    \noLine
    \UnaryInfC{\(\Theta(f) = \forall \vv{\alpha}:\mathcal{M}.\ \langle \tau_1, \dots, \tau_n \rangle \to \langle \tau_1', \dots, \tau_n' \mid \tau \rangle \)}

    \AxiomC{\( x \notin \dom(\Gamma')\)}
    \noLine
    \UnaryInfC{\([\vv{\beta}/\vv{\alpha}]\mathcal{M} \subseteq \mathcal{L}\)}
    \noLine

    \BinaryInfC{\(\tyjud{\Gamma, x:\rho, y_1:\rho_1', \dots, y_n:\rho_n'}{e : \xi}{\Gamma'}\)}
\RightLabel{\smaller (\textsc{T-Call})}
\UnaryInfC{\(\tyjud{\Gamma, y_1:\rho_1, \dots, y_n:\rho_n}{\letin{x}{f\langle\vv{\beta}\rangle(y_1, \dots, y_n)} e : \xi}{\Gamma'}\)}
\DisplayProof
\end{center}
\caption{Typing rules for standard expressions }
\label{fig:typing-standard}
\end{figure}

We explain our type system using \drvrule{T-Let} as a representative example.
The let expression allows an \emph{ownership transfer}.
The expression \(\letin{x}{y} e\) is well-typed only when the body \(e\) is typed under a type environment where \(x\) and \(y\) have types \(\tau_x\) and \(\tau_y\) that are obtained as split of \( \tau \), which is the type \( y \) originally had.
The split expressed by the type addition \(\tau = \tau_x + \tau_y\) (described below) intuitively means that some portion of \( y \)'s ownership is passed to its new alias \( x \).
The condition \( x \notin \dom(\Gamma') \) ensures that \( x \) does not escape its scope.

\emph{Type addition} \(\tau_1 + \tau_2\) is an essential feature in this type system to handle the ownership transfers.
The type addition is defined by the rules in \cref{fig:type-addition}.
\begin{figure}
\begin{center}
\AxiomC{\(\rule{0pt}{1.6ex}\)}
\RightLabel{\smaller (\textsc{A-Int})}
\UnaryInfC{\(\tint = \tint + \tint\)}
\DisplayProof
\hspace{2pt}
\AxiomC{\(\tau_1 = \tau_2\)}
\RightLabel{\smaller (\textsc{A-Ex})}
\UnaryInfC{\(\tau_2 = \tau_1\)}
\DisplayProof
\end{center}
\begin{center}
    \AxiomC{\(r_1 + r_2 \le 1\)}
    \AxiomC{\(s_1 + s_2 \le 1\)}
\RightLabel{\smaller (\textsc{A-Share})}
\BinaryInfC{\(\tref{\alpha}{r_1 + r_2}{\beta, s_1 + s_2} = \tref{\alpha}{r_1}{\beta, s_1} + \tref{\alpha}{r_2}{\beta, s_2}\)}
\DisplayProof
\end{center}
\begin{center}
    \AxiomC{}
\RightLabel{\smaller (\textsc{A-Borrow})}
\UnaryInfC{\(\tref{\alpha}{r+s}{\emptyset} = \tref{\alpha}{r}{\beta, s} + \tref{\beta}{s}{\emptyset}\)}
\DisplayProof
\end{center}
\caption{Type addition rules}
\label{fig:type-addition}
\end{figure}
The rules \drvrule{A-Share} and \drvrule{A-Borrow} express ownership distribution between two aliasing references but are conceptually different.
\emph{Sharing} distributes ownership between references of the same lifetime, whereas \emph{borrowing} is done between references of different lifetimes.
For example, \( \tref{\alpha}{1}{} = \tref{\alpha}{0.5}{\beta, 0.5} + \tref{\beta}{0.5}{} \) means that a reference with lifetime \( \alpha \) is lending half of its ownership to a reference with lifetime \(\beta\), where we have \( \alpha \lftgt \beta \) due to the well-formedness condition.
Type addition is extended to an operation on type environments, written \( \Gamma + \Gamma' \), in a pointwise manner.\footnote{If a variable \( x \) only appears in one of the type environment, say \( x \in \dom(\Gamma)\) but \( x \notin \dom(\Gamma') \), \( (\Gamma + \Gamma')(x) \) is given as \( \Gamma(x) \).}

Since we believe that most of the rules in~\cref{fig:typing-standard} are self-explanatory, we only explain a few nontrivial points.
In \drvrule{T-Var}, we allow some variables in the initial type environment or part of their ownership to be discarded.
This allows us to meet the condition \( x' \notin \dom(\Gamma') \) that appears in rules such as \drvrule{T-Let}.
The premise \( \Gamma \vdash o : \tint \) in \drvrule{T-Arith} expresses that each free variable \( x \) of an arithmetic expression \( o \) has type \( \tint \) in \( \Gamma \).
The condition \([\vv{\beta}/\vv{\alpha}]\mathcal{M} \subseteq \mathcal{L}\) in \drvrule{T-Call} means that the lifetime variables contained in the arguments' types must follow the function's constraint \(\mathcal{M}\) under the lifetime environment \(\mathcal{L}\).
As already explained, \( \tau_i \) and \( \tau'_i \) of a function type represent how the type of the arguments change by calling the function.
Types of \( y_i \), therefore, need to match \( \tau_i \) and \( \tau'_i \) (up-to substitution) before and after the function call, respectively.

\begin{figure}[t]
\begin{center}
    \AxiomC{\(\tyjud{\Gamma, x:\tref{\alpha}{1}{\emptyset}, y:\tint}{e : \rho}{\Gamma'}\)}
    \AxiomC{\( x \notin \dom(\Gamma')\)}
\RightLabel{\smaller (\textsc{T-MkRef})}
\BinaryInfC{\(\tyjud{\Gamma, y:\tint}{\letin{x}{\mkref y} e : \rho}{\Gamma'}\)}
\DisplayProof
\end{center}
\begin{center}
    \AxiomC{\(\tyjud{\Gamma, x:\tint, y:\tref{\alpha}{r}{B}}{e : \rho}{\Gamma'}\)}
\RightLabel{\smaller (\textsc{T-Deref})}
\UnaryInfC{\(\tyjud{\Gamma, y:\tref{\alpha}{r}{B}}{\letin{x}{\star y} e : \rho}{\Gamma'}\)}
\DisplayProof
\end{center}
\begin{center}
    \AxiomC{\(\tyjud{\Gamma, y:\tint}{e : \rho}{\Gamma'}\)}
    \AxiomC{\(\own_{\Gamma}(x) = 1\)}
\RightLabel{\smaller (\textsc{T-Assign})}
\BinaryInfC{\(\tyjud{\Gamma, y:\tint}{x := y; e : \rho}{\Gamma'}\)}
\DisplayProof
\end{center}
\caption{Typing rules for reference manipulations}
\label{fig:typing-reference}
\end{figure}

\cref{fig:typing-reference}~shows the typing rules for reference manipulations.
A newly created reference has the full ownership as expressed by \( x : \tref{\alpha}{1}{\emptyset} \) in the premise of \drvrule{T-MkRef}.
The rule \drvrule{T-Assign} has the condition \( \own(x) = 1 \) to ensure that only a reference with full ownership is updated.
Dereferencing can be done regardless of the ownership.
\footnote{If we want to disallow dereferencing references of ownership \( 0 \) we may add a premise \( r > 0 \) to \drvrule{T-Deref}; it is a matter of preference.}

\begin{figure}[t]
\begin{center}
    \AxiomC{\(\tau_x + \tau_y = \tref \alpha r B\)}
    \AxiomC{\(\tref \alpha r B = \rho_x + \rho_y\)}
    \noLine
    \BinaryInfC{\(\tyjud{\Gamma, x:\rho_x, y:\rho_y}{e : \rho}{\Gamma'}\)}
\RightLabel{\smaller (\textsc{T-Alias})}
\UnaryInfC{\(\tyjud{\Gamma, x:\tau_x, y:\tau_y}{\alias{x = y}; e :\rho}{\Gamma'}\)}
\DisplayProof
\end{center}
\begin{center}
  \AxiomC{\(\tyjudgement{\Theta}{\mathcal{M}}{\Gamma}{e : \rho}{\mathcal{L}'}{\Gamma'}\)}
  \AxiomC{\(\mathcal{M} = \mathcal{L} \cup \{\alpha\}\)}
  \AxiomC{\(\alpha = \min{\mathcal{M}}\)}
\RightLabel{\smaller (\textsc{T-NewLft})}
\TrinaryInfC{\(\tyjudgement{\Theta}{\mathcal{L}}{\Gamma}{\newlftin{\alpha} e : \rho}{\mathcal{L}'}{\Gamma'}\)}
\DisplayProof
\end{center}
\begin{center}
    \AxiomC{\(\tyjudgement{\Theta}{\mathcal{L}{\uparrow}_\alpha}{\Gamma{\uparrow}_\alpha}{e : \rho}{\mathcal{L}'}{\Gamma'}\)}
    \AxiomC{\(\alpha = \min{\mathcal{L}}\)}
\RightLabel{\smaller (\textsc{T-EndLft})}
\BinaryInfC{\(\tyjudgement{\Theta}{\mathcal{L}}{\Gamma}{\lftend \alpha ; e : \rho}{\mathcal{L}'}{\Gamma'}\)}
\DisplayProof
\end{center}
\caption{Typing rules for ghost instructions}
\label{fig:typing-ghost-inst}
\end{figure}

The rule \drvrule{T-Alias} allows us to transfer the ownership of references based on alias information.
Rules \drvrule{T-NewLft} and \drvrule{T-EndLft} are rules for lifetime introduction and termination, which do not exist in the type system of \consort{}.
The expression \(\newlftin{\alpha} e\) is typed when \(e\) is typed under the lifetime environment \(\mathcal{M}\) such that \(\mathcal{M} = \mathcal{L} \cup \{\alpha\}\) and \(\alpha = \min{\mathcal{M}}\), which mean that the newly introduced lifetime variable \(\alpha\) becomes a minimal lifetime in typing of \(e\).
The operator \(\uparrow_{\alpha}\) in \drvrule{T-EndLft} removes all information concerning \(\alpha\) from the environment and gives back the ownership that references with lifetime \( \alpha \) have been borrowing.
Concretely, \( \ctxterm{\Gamma}{\alpha} \) is defined by
\[
\begin{aligned}
  \ctxterm{\tint}{\alpha} &\defrel \tint \qquad \ctxterm{\tref{\beta}{r}{B}}{\alpha} \defrel
    \begin{cases}
        \tref{\beta}{r+s}{\emptyset} & (\text{if }B = (\alpha, s)) \\
        \tref{\beta}{r}{B} & (\text{otherwise})
    \end{cases} \\
    \ctxterm{\emptyset}{\alpha} &\defrel \emptyset \qquad
    \ctxterm{(\Gamma, x:\tau)}{\alpha} \defrel \begin{cases}
        \ctxterm{\Gamma}{\alpha} & (\text{if }\alpha = \lft(\tau)) \\
        \ctxterm{\Gamma}{\alpha}, x:(\ctxterm{\tau}{\alpha}) & (\text{if }\alpha \neq \lft(\tau))
    \end{cases} \\
\end{aligned}
\]
whereas \( \ctxterm{\mathcal L} \alpha \) is the subposet induced by \( \mathcal L \setminus \{ \alpha \}\).
The premise \( \alpha = \min{\mathcal{L}}\) says that we can only end a minimal lifetime.
This ensures that only a lifetime with which references are not lending their ownership can be ended.

The typing rules for functions and programs are defined in \cref{fig:typing-prog}.

\begin{figure}
\begin{center}
    \AxiomC{\(\Theta(f) = \forall \vv{\alpha}:\mathcal{L}.\ \langle \tau_1, \dots, \tau_n \rangle \to \langle \tau_1', \dots, \tau_n' \mid \tau \rangle\)}
    \AxiomC{\(\mathcal{L} \subseteq \vv{\alpha}\)}
    \noLine
    \BinaryInfC{\(\tyjudgement{\Theta}{\mathcal{L}}{x_1:\tau_1, \dots, x_n:\tau_n}{e : \tau}{\mathcal{L}}{x_1:\tau_1', \dots, x_n:\tau_n'}\)}
\RightLabel{\smaller (\textsc{T-FunDef})}
\UnaryInfC{\(\Theta \vdash f \mapsto \tuple{\vv{\alpha}}(x_1, \dots, x_n) e\)}
\DisplayProof
\end{center}
\begin{center}
    \AxiomC{\(\Theta \vdash d_1\)}
    \AxiomC{\(\cdots\)}
    \AxiomC{\(\Theta \vdash d_n\)}
    \noLine
    \TrinaryInfC{\(\dom(\Theta) = \dom(\lbrace d_1, \dots, d_n \rbrace)\)}

    \AxiomC{\(\tyjudgement{\Theta}{\emptyset}{\emptyset}{e : \tau}{\emptyset}{\emptyset}\)}
\RightLabel{\smaller (\textsc{T-Prog})}
\BinaryInfC{\(\vdash \langle \lbrace d_1, \dots, d_n \rbrace, e \rangle \)}
\DisplayProof
\end{center}
    \caption{Typing rules for functions and programs}
    \label{fig:typing-prog}
\end{figure}
Here, the rule \drvrule{T-Fundef} stipulates the contract that a function cannot have free variables and should end all lifetimes introduced in the function body by requiring \(\mathcal{L} \subseteq \vv{\alpha}\) and setting \(\mathcal{L}\) to both initial and ending lifetime context.
The rule \drvrule{T-Prog} checks that all the function definitions and the main expression are well-typed.
It also checks that the main expression does not have any free (lifetime) variables.

\cref{code:borfra-consort-demo} and \cref{code:borfra-rusthorn-demo} show the examples in \cref{code:consort-demo} and \cref{code:rusthorn} typed under our type system, respectively.
The typings are mostly the same as those of \consort{} and \rusthorn{}.
The difference between \cref{code:rusthorn} and \cref{code:borfra-rusthorn-demo} is that the reference \(x\) is not simply invalidated by the type system like in Rust but given a type of \(x:\tref{\alpha}{0}{\beta,1}\).
This demonstrates that our type system can express the lending references as a type, which is not possible in Rust.
We will utilize this difference in a later example.
Additionally, it is noteworthy that the typing of \cref{code:borfra-consort-demo} is obtained just by giving the identical lifetime to both references.
(Recall that \consort{} can be seen as a specific case of our type system.)

\begin{figure}
\begin{tabular}{ll}
\begin{minipage}[t]{0.47\hsize}
\begin{lstlisting}[style=mystyle]
let x = mkref 0 in
// $\color{Blue} x:\tref{\alpha}{1}{}$
let y = x in
// $\color{Blue} x:\tref{\alpha}{0}{},\ y:\tref{\alpha}{1}{}$
y := 1;
alias(x = y);
// $\color{Blue} x:\tref{\alpha}{0.5}{},\ y:\tref{\alpha}{0.5}{}$
assert( *x + *y = 2 )
\end{lstlisting}
    \caption{Typed example of \cref{code:consort-demo}}
    \label{code:borfra-consort-demo}
\end{minipage} &
\begin{minipage}[t]{0.47\hsize}
\begin{lstlisting}[style=mystyle]
let x = mkref 0 in
// $\color{Blue} x:\tref{\alpha}{1}{}$
let y = x in
// $\color{Blue} x:\tref{\alpha}{0}{\beta,1}, y:\tref{\beta}{1}{}$
y := 1;
endlft $\beta$;
// $\color{Blue} x:\tref{\alpha}{1}{}$
assert( *x = 1 )
\end{lstlisting}
    \caption{Typed example of \cref{code:rusthorn}}
    \label{code:borfra-rusthorn-demo}
\end{minipage}
\end{tabular}
\end{figure}

Another example of typing is given in \cref{code:typed-combination-intro}.
This example is the program in \cref{code:demo-combination-intro} typed under our type system (extended with pair types).
This program cannot be typed under \consort{} nor \rusthorn{} (namely, Rust) because \consort{} cannot insert \Alias{} statements after the use of \(z\) since \(z\) is a dynamic alias of either \(x\) or \(y\), and in Rust, we cannot touch \(x\) or \(y\) after the creation of \(p\) and \(q\).
Our type system can type this program by employing borrowing information and expressing partially lending references as a type of \(\tref{\alpha}{0.5}{\beta,0.5}\).

\begin{figure}
\begin{lstlisting}[style=mystyle]
// $\color{Blue}\Theta(\text{minmax})=\forall \alpha, \beta: \beta\lftlt\alpha.\ \tuple{\tref{\alpha}{1}{}, \tref{\alpha}{1}{}}$
// $\color{Blue}\qquad\qquad \to \tuple{\tref{\alpha}{0.5}{\beta,0.5}, \tref{\alpha}{0.5}{\beta,0.5} \mid \tuple{\tref{\beta}{0.5}{},\tref{\beta}{0.5}{}}}$
minmax(x, y) { if *x < *y then (x, y) else (y, x) }
// $\color{Blue}\Theta(\text{rand\_choose})=\forall \alpha, \beta: \beta\lftlt\alpha.\ \tuple{\tref{\alpha}{0.5}{\beta,0.5}, \tref{\alpha}{0.5}{\beta,0.5}}$
// $\color{Blue}\qquad\qquad \to \tuple{\tref{\alpha}{0}{}, \tref{\alpha}{0}{} \mid \tref{\alpha}{0.5}{\beta,0.5}}$
rand_choose(x, y) { if _ then x else y }
let x = mkref _ in           // $\color{Blue} x\COL\tref{\alpha}{1}{}$
let y = mkref _ in           // $\color{Blue} y\COL\tref{\alpha}{1}{}$
let (p, q) = minmax(x, y) in
// $\color{Blue} x\COL\tref{\alpha}{0.5}{\beta,0.5}, y\COL\tref{\alpha}{0.5}{\beta,0.5}, p\COL\tref{\beta}{0.5}{}, q\COL\tref{\beta}{0.5}{}$
let z = rand_choose(x, y) in // $\color{Blue} x\COL\tref{\alpha}{0}{}, y\COL\tref{\alpha}{0}{}, z\COL\tref{\alpha}{0.5}{\beta,0.5}$
assert( *p <= *z && *z <= *q );
endlft $\beta$;                    // $\color{Blue} \text{dispose } p, q;\ z\COL\tref{\alpha}{1}{}$
z := 1
\end{lstlisting}
    \caption{Typed example of \cref{code:demo-combination-intro}}
    \label{code:typed-combination-intro}
\end{figure}

\subsection{Type Preservation}
As usual, reduction preserves well-typedness, and, in our type system, well-typedness ensures the ownership invariant of fractional types:
\begin{theorem} \label{thm:ownsum-bound}
  Suppose that \( \vdash \tuple{D, e} \) and \(\tuple{\emptyset, \emptyset, \cdot,  e} \redmul D \langle H, R, \vv F, e' \rangle \).
  Then we have \( \tyjudgement \Theta {\mathcal L} {\Gamma} {e' : \tau} {\mathcal L'} {\Gamma'} \)  for some \( \Theta \), \( \Gamma \), \( \Gamma' \), \( \mathcal{L} \), \( \mathcal{L}' \) and \( \tau \).
  Moreover, for any \(a \in \dom(H)\), we have
  \(\sum_{x\in\setvar; R(x) = a}\own_{\Gamma}(x) \leq 1\).
  \qed
\end{theorem}
The proof of this theorem is given in \refappendix{sec:type-proof}.
In the theorem above, \( \Theta \), \( \mathcal L' \), \( \Gamma' \) and \( \tau \) are fixed, whereas \( \mathcal L \) and \( \Gamma \) differ at each step of the execution.\footnote{Strictly speaking, \( \mathcal L' \), \( \Gamma' \) and \( \tau \) also changes when a function is called.}
This is because, along the execution, we may introduce or terminate lifetime variables or change ownership of variables.
What is important, however, is that \( \Gamma \) always gives an assignment of ownership so that the sum of the ownerships of references pointing to an address \( a \) does not exceed \( 1 \).
This is a key property we exploit to show the soundness of the translation introduced in the next section.

\begin{remark} \label{rem:order-of-lifetime}
    We introduced the lifetime ordering because \cref{thm:ownsum-bound} does not hold if we have cyclic borrows.
    \cref{code:order-of-lifetime} shows an example of such programs.
    \qed
\end{remark}

\begin{figure}
\begin{lstlisting}[style=mystyle]
let x = mkref 0 in // $\color{Blue} x:\tref{\alpha}{1}{}$
let y = x in       // $\color{Blue} x:\tref{\alpha}{0}{\beta,1},\ y:\tref{\beta}{1}{}$
let z = y in       // $\color{Blue} y:\tref{\beta}{0}{\alpha, 1},\ z:\tref{\alpha}{1}{}$
endlft $\beta$           // $\color{Blue} x:\tref{\alpha}{\color{Red}1}{},\ z:\tref{\alpha}{\color{Red}1}{}$
// ownership sum of x and z becomes 2
\end{lstlisting}
    \caption{An example where cyclic borrows occur}
    \label{code:order-of-lifetime}
\end{figure}

\section{Translation} \label{section:translation}
This section introduces our translation from the source language into a standard functional language with non-deterministic assignments but without references.
Our translation is mainly based on \rusthorn{}~\cite{ref_rusthorn}.
We convert all references to a pair of current and prophecy values as described in \cref{section:intro} and exploit the prophecy technique when borrowings occur.
Employing this translation, we can verify that a well-typed program in the source language will not fail by checking the safety of the translated program using existing verifiers that can efficiently verify functional programs such as \toolname{Why3}~\cite{ref_why3} or \toolname{MoCHi}~\cite{ref_mochi}.

We first define our target language, then formulate translation rules, and finally state its soundness.

\subsection{Target Language}
\label{sec:target-lang}
In this subsection, we introduce the target language of our translation.

\subsubsection{Syntax}

The syntax of the target language is as follows:
\begin{gather*}
    o \mbox{ (arithmetic terms)} \Coloneqq\  n \mid x \mid o_1 \op o_2
\\[.2em]
    \begin{aligned}
        t \mbox{ (terms)} \Coloneqq\ & x \mid \tletin{x}{o} t \mid \tletin{x}{y} t \mid \tletin{x}{f(y_1, \dots, y_n)} t \\
        \mid\ & \tletin{x}{\tuple{y, z}} t \mid \tletin{x}{\fst y} t \mid \tletin{x}{\snd y} t \\
        \mid\ & \tletin{x}{\_} t \mid \tassume{x = y}; t \mid \tifz x \tthen t_1 \tels t_2 \mid \tfail
    \end{aligned}
\\[.2em]
    d \mbox{ (function definitions)} \ \Coloneqq\ f \mapsto (x_1, \dots, x_n) t
\\[.2em]
    P \mbox{ (programs)} \ \Coloneqq\ \langle \lbrace d_1, \dots, d_n \rbrace, t \rangle
\end{gather*}
We only describe the differences from the source language.
First, this language uses a \emph{non-deterministic value} \(\_\), whose exact value is not determined until a corresponding \(\textbf{assume}\) instruction appears.
For example, in \cref{code:non-deterministic-demo}, \(x\) and \(y\) can be any integer on their initialization.
But after running the third and fourth lines, we have \(x = 3\) and \(y = 3\).
Once the value of a non-deterministic value gets fixed, we can no longer change it.
So if we add \(\tassume{x = 5}\) at the end of this program, the program execution raises an error.
Second, this language is free of references (or lifetimes).
We no longer handle the heap memory in the operational semantics (described later), and thus, the verification problem for the target language is much more tractable than that for the source language.
Third, the target language has pairs, created by \(\tuple{x, y}\) and decomposed by \(\textbf{fst}\) and \(\textbf{snd}\).
Pairs will be used to model references using prophecy variables.

\begin{figure}
\begin{lstlisting}[style=mystyle]
let x = _ in  let y = _ in
assume(x = y);  assume(y = 3);
// `assume(x = 5)' here causes an error
\end{lstlisting}
    \caption{Example demonstrating non-determinism}
    \label{code:non-deterministic-demo}
\end{figure}

\newcommand{\contgt}[1]{\configtgt{\mathcal{D}}{\mathcal{S}}{\vv{\mathcal{F}}}{{#1}}}

\subsubsection{Operational Semantics}

We now introduce the operational semantics for our target language.
The set of runtime values \(\tgtval\) is defined recursively as follows:
\[
    \tgtval \ni v \Coloneqq n \mid \langle v_1, v_2 \rangle
\]

A configuration (runtime state) of this language has the form \(\contgt{t}\), consisting of function definitions \(\mathcal{D}\), a \emph{set} of registers \(\mathcal{S}\), a call stack \(\vv{\mathcal{F}}\), and the currently reducing term \(t\).
The main difference to the source language is that a configuration has a set of registers \(\mathcal{S}\) to handle the non-determinism.
Each register \(S \in \mathcal{S}\) is a partial mapping from \(\setvar\) to \(\tgtval\).
Important rules of the operational semantics are given by the rules in \cref{fig:target-opsem}; the rules for other constructs are standard and are defined similarly to those for the source language.
(See \refappendix{sec:operational-semantics-target} for the complete definition.)

\begin{figure}
    \begin{center}
        \AxiomC{\(\toml{S}' = \{S\{x \mapsto S(y)\} \mid S \in \toml{S}\}\)}
    \RightLabel{\smaller (\textsc{Rt-Let})}
    \UnaryInfC{\(\contgt{\tletin{x}{y}t} \longrightarrow \configtgt{\toml{D}}{\toml{S}'}{\vv{\toml{F}}}{t}\)}
    \DisplayProof
    \end{center}
    \begin{center}
        \AxiomC{\(\toml{S}' = \{S\{x \mapsto n\} \mid S \in \toml{S}, n \in \mathbb{Z}\}\)}
    \RightLabel{\smaller (\textsc{Rt-LetNondet})}
    \UnaryInfC{\(\contgt{\tletin{x}{\_}t} \longrightarrow \configtgt{\toml{D}}{\toml{S}'}{\vv{\toml{F}}}{t}\)}
    \DisplayProof
    \end{center}
    \begin{center}
    \AxiomC{\(\toml{S}' = \{S \in \toml{S} \mid S(x) = S(y)\}\)}
    \RightLabel{\smaller (\textsc{Rt-Assume})}
    \UnaryInfC{\(\contgt{\tassume{x = y}; t} \longrightarrow \configtgt{\toml{D}}{\toml{S}'}{\vv{\toml{F}}}{t}\)}
    \DisplayProof
    \end{center}
    \caption{Operational semantics of the target language (excerpt)}
    \label{fig:target-opsem}
\end{figure}

\subsection{Translation Rules}
\label{section:translation-rules}

\bgroup
\newcommand*{\PropOne}{TXZ} %
\newcommand*{\PropTwo}{PFZ} %

Here, we define our type-directed translation.
\cref{fig:trans-main} shows some notable translation rules; the translation is almost homomorphic for the other constructs (see \refappendix{sec:complete-translation} for details).
The translation relation is of the form \(\trans{\Gamma}{e\COL\tau}{t}\), meaning that \(e\) is translated to \(t\) under the type environment \( \Gamma \).
Note that if we ignore the  ``\( \Rightarrow t \)'' part, the rules are essentially the same as the typing rules given in \cref{section:type-system}.
We omitted the function type environment and the resulting environments ``\( \rhd \mathcal{L'} \mid \Gamma' \)'' for simplicity.
These environments are not needed to define the translation for the intraprocedual fragment of the source language.
To ease the presentation, we focus on this fragment in the rest of \cref{section:translation}.

\begin{figure}[t]
\begin{center}
    \AxiomC{\(\trans{\Gamma, x:\tref{\alpha}{1}{\emptyset}, y:\tint}{e : \rho}{t}\)}
\RightLabel{\smaller (\textsc{C-MkRef})}
\UnaryInfC{\(\trans{\Gamma, y:\tint}{\letin{x}{\mkref y} e : \rho}{\tletin{x}{\tuple{y, \_}} t}\)}
\DisplayProof
\end{center}
\begin{center}
    \AxiomC{\(\trans{\Gamma, x:\tint, y:\tref{\alpha}{r}{B}}{e : \rho}{t}\)}
    \AxiomC{\(r > 0\)}
\RightLabel{\smaller (\textsc{C-Deref-Pos})}
\BinaryInfC{\(\trans{\Gamma, y:\tref{\alpha}{r}{B}}{\letin{x}{\star y} e : \rho}{\tletin{x}{\fst y}} t\)}
\DisplayProof
\end{center}
\begin{center}
    \AxiomC{\(\trans{\Gamma, x:\tint, y:\tref{\alpha}{0}{B}}{e : \rho}{t}\)}
\RightLabel{\smaller (\textsc{C-Deref-Zero})}
\UnaryInfC{\(\trans{\Gamma, y:\tref{\alpha}{0}{B}}{\letin{x}{\star y} e : \rho}{\tletin{x}{\_}} t\)}
\DisplayProof
\end{center}
\begin{center}
  \AxiomC{\(\trans{\Gamma}{e : \rho}{t}\)}
  \AxiomC{\(\own_{\Gamma}(x) = 1\)}
\RightLabel{\smaller (\textsc{C-Assign})}
\BinaryInfC{\(\trans{\Gamma}{x := y; e : \rho}{\tletin{x}{\tuple{y, \snd x}}} t\)}
\DisplayProof
\end{center}
\begin{center}
  \AxiomC{\(\tau_x + \tau_y = \tref \alpha r B \quad \tref \alpha r B = \rho_x + \rho_y\)}
  \noLine
  \UnaryInfC{\(\trans{\Gamma, x:\rho_x, y:\rho_y}{e : \rho}{t} \quad t' = \convalias{\tau_x \shortrightarrow \rho_x}{\tau_x \shortrightarrow \rho_y}(x, y, t)\)}
\RightLabel{\smaller (\textsc{C-Alias})}
\UnaryInfC{\(\trans{\Gamma, x:\tau_x, y:\tau_y}{\alias{x = y}; e :\rho}{t'}\)}
\DisplayProof
\end{center}
\begin{center}
  \AxiomC{\(\ctxterm{\mathcal{L}}{\alpha} \mid \ctxterm{\Gamma}{\alpha} \vdash {e : \rho} \Rightarrow {t}\ \)}
  \AxiomC{\( \alpha = \min(\mathcal L)\)}
  \AxiomC{\(\Gamma\backslash\ctxterm{\Gamma}{\alpha} = \{x_1, \dots, x_n\}\)}
\RightLabel{\smaller (\textsc{C-EndLft})}
\TrinaryInfC{\(\trans{\Gamma}{\lftend \alpha; e : \rho}{\mathbf{assume}_{1\le i \le n}(\fst x_i = \snd x_i); t}\)}
\DisplayProof
\end{center}
\caption{Translation rules (excerpt)}
\label{fig:trans-main}
\end{figure}

As briefly explained, a reference \( x \) is translated into a pair \( \tuple{x, \proph x} \), where the first and the second elements (are supposed to) represent the current value of \( x \) and the \emph{future (or prophesized) value} of \( x \), respectively.
We treat the future value as if \(x\) has the value stored in \(x\) when \( x \)'s lifetime ends.
To help readers understand this idea, let us look at the rule \drvrule{C-MkRef}.
A new reference created by \( \mkref y \) is translated to a pair \( \tuple{y, \_} \).
The first element is \( y \) because the current value stored in the reference is the value \( y \) is bound to.
We set a non-deterministic value for the second element because, at the time of creation, we do not know what the value stored in this reference will be when its lifetime ends.
In other words, we randomly \emph{guess} the future value.
We \emph{check} if the guess was correct by confirming whether the first and second elements of the pair coincide when we drop this reference.
These checks are conducted by inserting \(\assume{\fst x_i = \snd x_i}\) that is obtained by translating \(\lftend \alpha\) as expressed in the rule \drvrule{C-EndLft}.
Here, \(x_i\) is a variable with the lifetime \( \alpha \).

Our translation obeys the following two principles: %
\begin{itemize}
  \item \emph{Trustable XOR Zero (\PropOne{})}: If a reference \( x \) has non-zero ownership, then the first element of the translation of \( x \) indeed holds the current value.
  \item \emph{Prophecy for Zero (\PropTwo{})}: The translation of a reference with zero ownership holds a prophesized value of a borrowing reference.
\end{itemize}

\PropOne{} resembles \consort{}'s idea that only a type with non-zero ownership can have a non-trivial refinement predicate.
The fact that we do not trust a reference with zero ownership is reflected in the rule \drvrule{C-Deref-Zero}, where dereferencing is translated into an assignment of a non-deterministic value.
\PropOne{} enables us to express an update of a reference with full ownership as an update of the first element of the pair as in \drvrule{C-Assign}.
(Note that, by \cref{thm:ownsum-bound}, aliases of a reference with ownership \( 1 \) must have ownership \( 0 \).)

\PropTwo{} is the core of the \rusthorn{}-style prophecy technique.
A reference of ownership \( 0 \) can regain some ownership in two ways: by retrieving the lending ownership through \texttt{endlft} or by ownership transfer through \Alias{} annotations.
\PropTwo{} handles the former case. (How the latter is handled will be explained below.)
Suppose a reference \( x \) is lending its ownership to \( y \).
\PropTwo{} says that \( \fst x = \snd y \) in the translated program.
When the lifetime of \( y \) ends, in the translated program, the current value of \( y \) must equal the prophesized value because of \( \assume{\fst y = \snd y } \).
Hence, at this moment, we have \( \fst x = \snd y = \fst y \), which means that \( x \) also holds the correct value.
Therefore, \PropOne{} is not violated when the ownership of \( x \) becomes non-zero.

We now explain the rule~\drvrule{C-Alias}, which is the key rule to understand how our translation deals with ownership transfers.
The rule \drvrule{C-Alias} uses an auxiliary function \(\convalias{\tau_x \shortrightarrow \rho_x}{\tau_y \shortrightarrow \rho_y}(x, y, t)\) defined as follows:
\[
  \convalias{\tau_x \shortrightarrow \rho_x}{\tau_y \shortrightarrow \rho_y}(x, y, t) \defrel \begin{cases}
    t & (\own(\tau_x) \!=\! \own(\rho_x) \!=\! 0) \\[.4em]
    \letin{x}{\tuple{\fst y, \snd x}} t & \left(\let\scriptstyle\textstyle\substack{\own(\tau_y) \\ \own(\rho_x), \own(\rho_y)} > 0\right) \\[.4em]
    \hspace{-0.5em}\left(
    \begin{array}{l}
        \letin{y}{\tuple{\fst x, \snd y}} \\
        \letin{x}{\tuple{\snd y, \snd x}} \\
        t
    \end{array}
    \right. & \left(\let\scriptstyle\textstyle\substack{\own(\tau_x) > 0 \\ \text{and}\vspace{.1em} \\ \own(\rho_x) = 0}\right) \\[2.0em]
    \convalias{\tau_y \shortrightarrow \rho_y}{\tau_x \shortrightarrow \rho_x}(y, x, t) & (\text{otherwise})
  \end{cases}
\]
The essential case is when \(\own(\tau_x) > 0 \)  and \( \own(\rho_x) = 0\), that is, when \(y\) takes all the ownership of \(x\).
In this case, \(x\) takes the future value of \(y\) to obey \PropTwo{}.
At the same time, to ensure \PropOne{}, \(y\) takes the old value of \( x \) so that \(y\) has the correct value. (Note that \(\own(\tau_y)\), the previous ownership of \(y\), can be zero.)
Similarly, when \(\own(\tau_y), \own(\rho_x), \own(\rho_y) > 0\), the value \( y \) holds is passed to \( x \) to ensure \PropOne{} (\(\own(\rho_y) > 0\) is necessary to avoid case overlapping.)
\cref{code:borfra-trans-consort} and \cref{code:borfra-trans-rusthorn} show the translated example of \cref{code:borfra-consort-demo} and \cref{code:borfra-rusthorn-demo}.

\begin{figure}
\begin{tabular}{ll}
\begin{minipage}[t]{0.47\hsize}
\begin{lstlisting}[style=mystyle]
let x = (0, _) in
let y = (fst x, _) in
let x = (snd y, snd x) in
let y = (fst y + 1, snd y) in
let x = (fst y, snd x) in
assert(fst x + fst y = 2)
\end{lstlisting}
    \caption{Translated example of \cref{code:borfra-consort-demo}}
    \label{code:borfra-trans-consort}
\end{minipage} &
\begin{minipage}[t]{0.47\hsize}
\begin{lstlisting}[style=mystyle]
let x = (0, _) in
let y = (fst x, _) in
let x = (snd y, snd x) in
let y = (1, snd y) in
assume(fst y = snd y);
assert(fst x = 1)
\end{lstlisting}
    \caption{Translated example of \cref{code:borfra-rusthorn-demo}}
    \label{code:borfra-trans-rusthorn}
\end{minipage}
\end{tabular}
\end{figure}

\begin{remark}[Difference from \rusthorn{}]
The main difference between our translation and \rusthorn{}'s method is that our method converts all references to pairs, whereas \rusthorn{} only converts mutably borrowing references.
This gap emerges because our type system enables an immutable reference to recover its ownership by \texttt{alias}; meanwhile, in Rust, once a mutable reference becomes immutable, it cannot be changed to mutable.
Note that in \cref{code:borfra-trans-rusthorn}, both references \(x\) and \(y\) become pairs unlike in \rusthorn{} (\cref{code:rusthorn-red}).
\end{remark}
\egroup

\subsection{Soundness}
Our translation is sound, that is, if the translated program does not reach \( \mathbf{Fail} \) neither does the original program.

As we are considering the intraprocedual fragment, we omit the stacks from configurations.
For clarity, we write the reduction in the source language and the target language as \(\tuple{H, R, e} \redsrc \tuple{H', R', e'}\) and \(\tuple{\toml{S}, t} \redtgt \tuple{\toml{S}', t'}\), respectively.
\newcommand{\simptranslation}[4]{{#1} \mid {#2} \vdash {#3} \Rightarrow {#4}}

\begin{theorem}[Soundness of the translation] \label{thm:translation-soundness}
    Assume \(\simptranslation{\emptyset}{\emptyset}{e:\tau}{t}\) and \(\tuple{\{\emptyset\}, t} \longarrownot\redmultgt \mathbf{Fail}\), then \(\tuple{\emptyset, \emptyset, e} \longarrownot\redmulsrc \mathbf{Fail}\) holds.
\end{theorem}
We briefly explain the proof strategy below; the proof is given in \refappendix{sec:soundness-proof}.
We define an (indexed) simulation \( \simrel_{\mathcal{L}, \Gamma} \) between the configurations of the two languages.
The relation \( \tuple{H, R, e} \simrel_{\mathcal{L}, \Gamma} \tuple{\toml{S}, t}\) not only requires \(\trans{\Gamma}{e\COL\tau}{t}\) to hold, but also requires that TXZ and PFZ hold; without them the relation is too weak to be an invariant.
Whenever \( \tuple{H, R, e} \redsrc \tuple{H', R', e'}\) we show that there exists \( \tuple{\toml{S'}, t'}\) such that \( \tuple{\toml{S}, t}\redmultgt \tuple{\toml{S'}, t'} \) and \( \tuple{H', R', e'} \simrel_{\mathcal{L'}, \Gamma'} \tuple{\toml{S'}, t'}\) \emph{for some \( \mathcal{L'}, \Gamma' \)}.
Such environments can be chosen by an argument identical to the proof of type preservation (\cref{thm:ownsum-bound}).

\section{Preliminary Experiments} \label{section:experiments}
We conducted preliminary experiments to evaluate the effectiveness of our translation.
In this experiment, we first typed the benchmark programs and translated them into OCaml programs by hand\footnote{This translation process can be automated if we require programmers to provide type annotations or if the type inference mentioned in Section~\ref{section:conclusion} is worked out.} (with slight optimizations).
After that, we measured the time of the safety verification of the translated programs using \toolname{MoCHi}~\cite{ref_mochi}, a fully-automated verifier for
OCaml programs.
The experiments were conducted on a machine with AMD Ryzen 7 5700G 3.80 GHz CPU and 32GB RAM.
We used \toolname{MoCHi} of build \texttt{a3c7bb9d} as frontend solver, which used Z3~\cite{ref_z3} (version 4.11.2) and HorSat2~\cite{ref_horsat2} (version 0.95) as backend solvers.

We ran experiments on small but representative programs with mutable references and summerized the results in \cref{tbl:experiment}.
Columns `program' and `safety' are the name of programs and whether they involve \(\fail\) in assertions, respectively.
When two programs have the same name but different safeties, they represent slightly modified variations of a single program with different safety.
Columns `ConSORT?' and `RustHorn?' indicate that the program can be typed by ConSORT or RustHorn, respectively.
Columns \(B_o\) and \(B_t\) are the byte counts of the original and translated programs, respectively.
The column `time(sec)' shows how many seconds \toolname{MoCHi} takes to verify translated programs.
These results show that our method works in practical time for representative programs, though it increases the program size by 2 to 4 by translation.

Below, we describe each benchmark briefly.
`consort-demo,' `rusthorn-demo,' and `minmax' are the sample programs in \cref{code:consort-demo,code:rusthorn,code:demo-combination-intro}.
`simple-loop' and `shuffle-in-call', both included in \consort{}'s benchmark, are the programs utilizing fractional ownership.
`just-rec' and `inc-max' are from \rusthorn{}, which exploits borrowable ownership.
`linger-dec' is the program with loop and dynamic allocating of references and adopted by both \consort{} and \rusthorn{} benchmark.
`hhk2008' is a \toolname{SeaHorn}~\cite{ref_seahorn} test to check whether the verifier can find an invariant on the loop.
See \refappendix{sec:benchmark-programs} for benchmark programs.

\begin{table}[hbtp]
  \centering
  \caption{Experimental results}
    \label{tbl:experiment}
  \begin{tabular}{|cccc|ccc|c|}
    \hline
    program & safety & \hspace*{1pt} ConSORT? \hspace*{1pt} & \hspace*{1pt} RustHorn? \hspace*{1pt} & \(\quad B_o \quad \) & \(\ B_t\ \) & \hspace*{1pt} \(B_t / B_o\) \hspace*{1pt} & \hspace*{2pt} time(sec) \hspace*{2pt} \\
    \hline\hline

    consort-demo & safe & \checkmark & \(\times\) & \(76\) & \(285\) & \(3.75\) & \(<0.1\) \\
    
    rusthorn-demo & safe & \checkmark & \checkmark & \(68\) & \(276\) & \(4.06\) & \(0.8\) \\
    borrow-merge & safe & \(\times\) & \(\times\) & \(216\) & \(756\) & \(3.5\) & \(1.0\) \\
    \hline
    \multirow{2}{*}{simple-loop} & safe & \multirow{2}{*}{\checkmark} & \multirow{2}{*}{\checkmark} & \(194\) & \(347\) & \(1.79\) & \(<0.1\) \\
     & unsafe &  &  & \(150\) & \(320\) & \(2.13\) & \(0.2\) \\
    \hline
    \multirow{2}{*}{just-rec} & safe & \multirow{2}{*}{\checkmark} & \multirow{2}{*}{\checkmark} & \(156\) & \(362\) & \(2.32\) & \(0.6\) \\
     & unsafe &  &  & \(157\) & \(372\) & \(2.37\) & \(0.2\) \\
    \hline
    \multirow{2}{*}{shuffle-in-call} & safe & \multirow{2}{*}{\checkmark} & \multirow{2}{*}{\(\times\)} & \(121\) & \(395\) & \(3.26\) & \(0.8\) \\
     & unsafe &  &  & \(121\) & \(395\) & \(3.26\) & \(0.2\) \\
    \hline
    \multirow{2}{*}{inc-max} & safe & \multirow{2}{*}{\(\times\)} & \multirow{2}{*}{\checkmark} & \(143\) & \(550\) & \(3.85\) & \(1.2\) \\
     & unsafe &  &  & \(143\) & \(549\) & \(3.84\) & \(0.4\) \\
    \hline
    \multirow{2}{*}{minmax} & safe & \multirow{2}{*}{\(\times\)} & \multirow{2}{*}{\(\times\)} & \(251\) & \(920\) & \(3.67\) & \(0.2\) \\
     & unsafe &  &  & \(251\) & \(920\) & \(3.67\) & \(0.6\) \\
    \hline
    \multirow{2}{*}{linger-dec} & safe & \multirow{2}{*}{\checkmark} & \multirow{2}{*}{\checkmark} & \(239\) & \(836\) & \(3.5\) & \(1.1\) \\
     & unsafe &  &  & \(243\) & \(841\) & \(3.46\) & \(0.3\) \\
    \hline
    hhk2008 & safe & \checkmark & \checkmark & \(321\) & \(630\) & \(1.96\) & \(2.7\) \\
    \hline
    \end{tabular}
\end{table}

\section{Related Work} \label{section:related}
As introduced in \cref{section:intro}, \consort{}~\cite{ref_consort} and \rusthorn{}~\cite{ref_rusthorn} are the direct ancestor of this work.
We have combined the two approaches as borrowable fractional ownership types.

\toolname{Creusot}~\cite{ref_creusot} and \toolname{Aeneas}~\cite{ref_aeneas} are verification toolchains for Rust programs based on a translation to functional programs. %
\toolname{Creusot} removes references and translates Rust programs to \toolname{Why3}~\cite{ref_why3} programs using \rusthorn{}'s prophecy technique.
Our translation is closer to that of \toolname{Creusot} than that of \rusthorn{} in that the target is a functional program rather than CHCs.
But very unlike \toolname{Creusot}, our translation accommodates the \consort{}-style fractional ownership and alias annotations.
\toolname{Aeneas} uses a different encoding to translate Rust programs to pure functional programs for interactive verification (in Coq, F*, etc.).
However, they can't support advanced patterns like nested borrowing because their model of a mutable reference is not a first-class value, unlike RustHorn's prophecy-based approach.

There are other verification methods and tools that utilize some notion of ownership for imperative program verification, such as \toolname{Steel}~\cite{ref_steel,ref_steelcore}, \toolname{Viper}~\cite{ref_viper}, and \toolname{RefinedC}~\cite{ref_refinedc}.
Still, their approaches and design goals are quite different from ours in that they do \emph{semi-}automated program verification in F* or low-level separation logic, requiring more user intervention, such as annotations of loop invariants and low-level proof hints about ownership.

\toolname{SeaHorn}~\cite{ref_seahorn} and \toolname{JayHorn}~\cite{ref_jayhorn} both introduce a fully automated verification framework for programs with mutable references.
\toolname{SeaHorn} is for LLVM-based languages (C, C++, Rust, etc.) and \toolname{JayHorn} is for Java.
Both do not use ownership types but model the heap memory directly as a finite array. %
As a result, they are ineffective or imprecise for programs with dynamic memory allocation, as shown in the experiments of \consort{} and \rusthorn{}.

The work by Suenaga and Kobayashi~\cite{ref_suenaga_09} and \toolname{Cyclone}~\cite{ref_cyclone} used ownership types to ensure memory safety for a language with explicit memory deallocation.
We expect that our borrowable fractional ownership types can also be used to improve their methods.

\section{Conclusion and Future Work} \label{section:conclusion}
We have presented a type system based on the new notion of borrowable fractional ownership, and a type-directed translation to reduce the verification of imperative programs to that of programs without mutable references. 
Our approach combines that of \consort{}~\cite{ref_consort} and \rusthorn~\cite{ref_rusthorn},
enabling automated verification of a larger class of programs.

Future work includes automated type inference and
an extension of the type system to allow nested references.
Type inference would not be so difficult if we assume some lifetime annotations as in Rust.
Concerning nested references, as a naive extension of reference types from \(\tref{\alpha}{r}{B}\) to \(\tau\rawtref{\alpha, r}{B}\) seems too restrictive, we plan to introduce a reference type of the form \(\xi\backslash\tau/\rho\rawtref{\alpha,r}{\beta,s}\), where \(\xi\) is the borrowing type, \(\tau\) is the current content type, and \(\rho\) is the lending type to others.
The point is that, instead of just keeping the amount of ownership being borrowed, we should keep in what type a reference is being borrowed.

\subsubsection{Acknowledgements}
\begin{sloppypar}
We would like to thank anonymous referees for their useful comments.
This work was supported by JSPS KAKENHI Grant Number JP20H05703 and JP22KJ0561.
\end{sloppypar}

\bibliographystyle{splncs04}
\bibliography{mybib}

\clearpage
\appendix
\newcommand{\authcount}[1]{\relax}
\startcontents[sections]
\noindent{\Large\textbf{Appendicies}}
\vspace{2em}
\printcontents[sections]{l}{1}{\setcounter{tocdepth}{2}}
\section{Proof of Type Preservation (\cref{thm:ownsum-bound})}
\label{sec:type-proof}
This section first defines the typing relation of configurations, prepares some lemmas, and then proves the type preservation.

\newcommand*{\brr}{\mathbf{Brr}}
\newcommand*{\bby}{\mathbf{BBy}}
\newcommand*{\bfrm}{\mathbf{BFrm}}
\newcommand*{\bto}{\curvearrowright}
\newcommand*{\lift}[1]{{\uparrow}_{#1}}
\newcommand*{\tyenv}{\Gamma}
\newcommand*{\genv}{\Delta} %
\newcommand*{\totenv}{\Gamma_\mathrm{tot}}
\newcommand*{\lenv}{\mathcal{L}}
\newcommand*{\fenv}{\Theta}
\newcommand*{\cons}[4]{#1 \mid #2 \vdash \tuple{#3, #4}}
\newcommand*{\adr}{a}
\newcommand*{\lvar}{\alpha}
\newcommand*{\lvarTwo}{\beta}
\newcommand*{\lvarThree}{\gamma}

\newcommand*{\ty}{\tau}
\newcommand*{\tyTwo}{\rho}
\newcommand*{\upd}[2]{\{#1 \mapsto #2\}}
\newcommand*{\hole}{[\cdot]}
\newcommand*{\ejdg}[8]{#1 \mid \hole : (#2 \rhd #3 \mid #4) \vdash #5 : #6 \rhd #7 \mid #8}

\newcommand*{\sub}[3]{[#1/ #2]#3}
\newcommand*{\seq}[1]{\tilde #1}
\subsection{Configuration Typing}
Here, we explain how a configuration is typed.
The typing rule for configurations is as follows:
\begin{mathpar}
  \inferrule{%
    \fenv \vdash D \\
    \ejdg \fenv {\ty_i} {\lenv_i} {\tyenv_i} {F_i} {\ty_{i-1}} {\lenv_{i-1}} {\tyenv_{i-1}} \quad (1 \le i \le n)\\\\
    \totenv = \tyenv + \genv \\
    \cons \lenv \totenv H R \\
    \tyjudgement{\Theta} \lenv \tyenv {e : \ty_n} {\lenv_n}{\tyenv_n}
  }
  {\vdash^D \configsrc H R {F_n : \cdots : F_1} e}
{\smaller (\textsc{T-Conf})}
\end{mathpar}
Premises \( \fenv \vdash D \), \( \ejdg \fenv {\ty_i} {\lenv_i} {\tyenv_i} {F_i} {\ty_{i-1}} {\lenv_{i-1}} {\tyenv_{i-1}} \) and \( \tyjudgement{\Theta} \lenv \tyenv {e : \ty_n} {\lenv_n}{\tyenv_n} \) just say that the function definitions, call stack, and currently reducing expression are all well-typed.
The notable part of this rule is \( \cons \lenv \totenv H R \), which intuitively expresses that the current heap and registers are consistent with the information implied by the types in \( \totenv \).
For example, if \( \cons \lenv \totenv H R \) and \( x : \tref \lvar {0.5} {\lvarTwo, 0.5} \in \totenv \), then \( R(x) \) must be an address that is in \( \dom(H)\) and there must be an aliasing reference pointing to \( R(x) \) with a lifetime \( \lvarTwo \).
The addition \( \totenv = \tyenv + \genv \) represents the fact that we also keep track of the information of variables that may have been discarded because of the rule \drvrule{T-Var}.

The rest of this subsection is devoted to the explanation of the judgments that have not been introduced so far, namely \( \cons \lenv \totenv H R\) and \( \ejdg \fenv {\ty_i} {\lenv_i} {\tyenv_i} {F_i} {\ty_{i-1}} {\lenv_{i-1}} {\tyenv_{i-1}} \).
\begin{figure}[t]
  \begin{align*}
    \blft(\tref \lvar r B) &\defrel
    \begin{cases}
        \beta & (B = (\beta, s)) \\
        \text{undefined} & (B = \emptyset)
    \end{cases}
    \qquad \blft(\tint) \defrel \text{undefined} \\
    \bown(\tref \lvar r B) &\defrel
    \begin{cases}
        s & (B = (\beta, s)) \\
        0 & (B = \emptyset)
    \end{cases}
  \qquad \bown(\tint) \defrel 0 \\
  \\
  \brr^R_\tyenv(\lvar \bto \lvarTwo; \adr) &\defrel \sum \{\bown(\ty) \mid x : \ty \in \tyenv, \lft(\ty) = \lvar, \blft(\ty) = \lvarTwo, R(x) = \adr \} \\
  \bby^R_\tyenv(\lvar, \adr) &\defrel \sum \{\bown(\ty) \mid x : \ty \in \tyenv, \blft(\ty) = \lvar, R(x) = \adr \} \\
  \bfrm^R_\tyenv(\lvar, \adr) &\defrel \sum \{\bown(\ty) \mid x : \ty \in \tyenv, \lft(\ty) = \lvar, R(x) = \adr \}\\
  \\
  \own^R_\tyenv(\lvar, \adr) &\defrel \sum \{ \own(\ty) \mid x : \ty \in \tyenv, \lft(\ty) = \lvar, R(x) = \adr \} \\
  \own^R_\tyenv(\adr) &\defrel \sum \{ \own(\ty) \mid x : \ty \in \tyenv, R(x) = \adr \}
\end{align*}
\caption{Operations for calculating ownership. (Here the set-builder notation should be interpreted as a multiset.)}
\label{fig:borrow-notations}
\end{figure}

\subsubsection{Heap and register typing:}
We now formally define the relation \( \cons \lenv \totenv H R \).
Notations for ownership used in the definition are defined in \cref{fig:borrow-notations}.
\begin{definition}
  We say that \emph{a heap \( H \) and a register \( R \) are consistent with a lifetime environment \( \lenv \) and a type environment \( \tyenv \)}, written \( \cons \lenv \tyenv H R \), if \( \lenv \vdash_\mathrm{WF} \tyenv \) and the following conditions hold:
  \begin{description}
    \item[(fraction consistency)] for every \( \adr \in \dom(H) \), \( \own^R_\tyenv(a) \le 1 \).
    \item[(borrow consistency)] for every \( \adr \in \dom(H) \) and \( \lvar \in \lenv \), \(  \bby^R_\tyenv(\lvar; \adr) \le \own^R_\tyenv(\lvar; \adr) + \bfrm^R_\tyenv(\lvar; \adr) \).
    \item[(type consistency)] for each \( x : \ty \in \tyenv \), if \( \ty = \tint \) then \( R(x) \in \mathbb Z \), and otherwise, \( R(x) \in \addr \).
    \item[(memory consistency)] for every \( x : \ty \in \tyenv \) such that \( \ty \) is a reference type, \( R(x) \in \dom(H) \) and \( H(R(x)) \in \mathbb Z\).
  \end{description}
\end{definition}
Fraction consistency says that the sum of the ownership of references pointing to an address \( \adr \) cannot exceed 1.
This is the property we want to prove.
On the other hand, borrow consistency is just a technical invariant that is exploited in the proof.
Roughly speaking, it expresses the fact that an ownership that is borrowed cannot be discarded.
Borrow consistency will be used to ensure that the ownership that was borrowed will be given back properly.

\subsubsection{Return context typing:}
The typing rule for return context is given as
\begin{mathpar}
\inferrule{ \tyjudgement \fenv \lenv {\tyenv, x : \ty} {e : \tyTwo} {\lenv'} {\tyenv'} \and x \notin \dom(\tyenv')}
{\ejdg \fenv \ty \lenv \tyenv {\letin x \hole e} {\tyTwo} {\lenv'} {\tyenv'}}
{\smaller (\textsc{TC-Let)}}
\end{mathpar}
As usual, the judgment \( \ejdg \fenv \ty \lenv \tyenv F {\tyTwo} {\lenv'} {\tyenv'} \) describes that if the hole has a type \( \ty \) under a type environment \( \tyenv \), then the context \( F \) has a type \( \tyTwo \) under \( \tyenv \).
\subsection{Auxiliary Lemmas}
We now prepare some lemmas that are used to prove type preservation.
Some of the lemmas are not minor propositions that are used to help prove the theorem, but rather (essentially) are some cases of the case analysis made in the proof of the theorem.
We have separated these cases as lemmas for clarity.

\subsubsection{Basic properties of the type system}
\begin{lemma}[Weakening]
  \label{lem:weakening}
  Let \( \tyjudgement \fenv \lenv \tyenv {e : \tyTwo} {\lenv'} {\tyenv'} \).
  If \( x \notin \dom(\tyenv) \) and \( \ty \) is of the form \( \tint \) or \( \tref {\lvar} 0 {} \) for \( \lvar \in \lenv\), then we have \( \tyjudgement \fenv \lenv {\tyenv, x : \ty} {e : \tyTwo} {\lenv'} {\tyenv'} \).
\end{lemma}
\begin{proof}
  By straightforward induction on the structure of the type derivation.
  \qed
\end{proof}
\begin{lemma}[Weakening2]
  \label{lem:weakening2}
  Let \( \tyjudgement \fenv \lenv \tyenv {e : \tyTwo} {\lenv'} {\tyenv'} \).
  If \( \lenv \vdash_\mathrm{LF} \tyenv + \tyenv''  \) and \( \lenv' \vdash_\mathrm{LF} \tyenv' + \tyenv''  \), then we have \( \tyjudgement \fenv \lenv {\tyenv + \tyenv''} {e : \tyTwo} {\lenv'} {\tyenv' + \tyenv''} \).
\end{lemma}
\begin{proof}
  By straightforward induction on the structure of the type derivation.
  \qed
\end{proof}
\
\begin{lemma}[Weakening of lifetime variables]
  \label{lem:weakening-lvar}
  Let \( \tyjudgement \fenv \lenv \tyenv {e : \tyTwo} \lenv {\tyenv'} \).
  If \( \lenv \subseteq \lenv' \),\footnote{Here the inclusion means the existence of an injection from \( \lenv \) to \( \lenv' \) that preserves and reflects the strict order.} then we have \( \tyjudgement \fenv {\lenv'} {\tyenv} {e : \tyTwo} {\lenv'} {\tyenv'} \).
\end{lemma}
\begin{proof}
  By straightforward induction on the structure of the type derivation.
  \qed
\end{proof}

\begin{lemma}[Substitution]
  \label{lem:substitution}
  If \( \tyjudgement \fenv \lenv \tyenv {e : \ty} {\lenv'} {\tyenv'} \) and \( x' \notin \dom(\tyenv) \), then \( \tyjudgement \fenv \lenv {\sub {x'} x \tyenv} {\sub {x'} x e : \ty} {\lenv'} {\sub {x'} x \tyenv'} \).
\end{lemma}
\begin{proof}
  By straightforward induction on the structure of the type derivation.
  \qed
\end{proof}
\begin{lemma}[Substitution of lifetime variables]
  \label{lem:substitution-lvar}
  If \( \tyjudgement \fenv \lenv \tyenv {e : \ty} {\lenv'} {\tyenv'} \) and \( \lvarTwo \notin \lenv \), then \( \tyjudgement \fenv {\sub \lvarTwo \lvar \lenv} {\sub \lvarTwo \lvar \tyenv} {e : \sub \lvarTwo \lvar \ty} {\sub \lvarTwo \lvar {\lenv'}} {\sub \lvarTwo \lvar \tyenv'} \).
\end{lemma}
\begin{proof}
  By straightforward induction on the structure of the type derivation.
  \qed
\end{proof}

\subsubsection{Lemmas about \texorpdfstring{\( \cons \lenv \tyenv H R \)}{consistency}}
\begin{lemma}
  \label{lem:sum-own-and-brr}
  Suppose that \( \lenv \vdash_\mathrm{WF} \tyenv \).
  Given a heap \( H \) and a register \( R \), for each \( \adr \in \dom(H) \), we have
\begin{enumerate}
  \item \( \bby^R_\tyenv(\lvar; \adr) = \sum_{\lvarTwo \in \lenv} \brr_\tyenv^R(\lvar \bto \lvarTwo; \adr) \) for every \( \lvar \in \lenv \),
  \item \( \bfrm^R_\tyenv(\lvar; \adr) = \sum_{\lvarTwo \in \lenv} \brr_\tyenv^R(\lvarTwo \bto \lvar; \adr) \) for every \( \lvar \in \lenv \), and
  \item \( \own^R_\tyenv(\adr) = \sum_{\lvar \in \lenv} \own_\tyenv^R(\lvar; \adr) \) for each \( \adr \in \dom(H) \).
\end{enumerate}
\end{lemma}
\begin{proof}
  Obvious from the definitions. (The only condition of \( \lenv \vdash_\mathrm{WF} \tyenv \) we use is the type consistency, so the assumption is unnecessarily strong; this will, however, not cause any problem.)
  \qed
\end{proof}

The following two lemmas are about lifetime termination.
\cref{lem:lift-preserves-cons} states that the consistency relation is preserved under lifetime termination, and is the key lemma to show the type preservation under the reduction caused by \drvrule{RS-Endlft}.

\begin{lemma}
  \label{lem:borrow-cons-min}
  Suppose that \( \cons \lenv \tyenv H R \) and \( \lvar \in \min(\lenv)  \).
  Then we have \( \bby^R_\tyenv(\lvar; \adr) \le \own^R_\tyenv(\lvar; \adr) \).
\end{lemma}
\begin{proof}
  Because of the well-formedness condition and the minimality of \( \lvar \), we must have \( \bfrm_\tyenv^R(\lvar; \adr) = 0 \) for all \( \adr \in \dom(H)\).
  Thus, we have \( \bby^R_\tyenv(\lvar; \adr) \le \own^R_\tyenv(\lvar; \adr) \) by the borrow consistency.
  \qed
\end{proof}

\begin{lemma}
  \label{lem:lift-preserves-cons}
  Suppose that \( \cons \lenv \tyenv H R \) and \( \lvar \in \min(\lenv)\).
  \begin{enumerate}
    \item We have \( \brr^R_\tyenv(\lvarTwo \bto \lvarThree; a ) = \brr^R_{\tyenv \lift \lvar }(\lvarTwo \bto \lvarThree; a) \) for every \( \adr \in \dom(H) \) and \( \lvarTwo, \lvarThree \in \lenv \lift \lvar \).
      \label{it:lem:lift-preserves-cons:brr}
    \item For all \( \adr \in \dom(H) \) and \( \lvarTwo \in \lenv \lift \lvar \),
      \begin{equation}
      \own^R_{\tyenv \lift \lvar}(\lvarTwo; \adr) = \own_{\tyenv}^{R}(\lvarTwo; \adr) + \brr_\tyenv^R(\lvarTwo \bto \lvar; \adr). \label{eq:lem:lift-preserves-cons:own-with-lft}
      \end{equation}

      Hence, for any \( \adr \in \dom(H) \), we have
      \begin{equation}
      \own^R_{\tyenv \lift \lvar}(\adr) = \sum_{\beta \neq \alpha }\own_{\tyenv}^{R}(\lvarTwo; \adr) + \bby_\tyenv^R(\lvar; \adr). \label{eq:lem:lift-preserves-cons:own}
      \end{equation}
    \label{it:lem:lift-preserves-cons:own-plus-brr}
    \item We have \(\cons {\lenv \lift \lvar} {\tyenv \lift \lvar} H R \).
    \label{it:lem:lift-preserves-cons:main}
  \end{enumerate}
\end{lemma}
\begin{proof}
  The statement~\ref{it:lem:lift-preserves-cons:brr} is a straightforward consequence of the definition of \( \brr \) and \( \tyenv \lift \lvar \).

  (Proof of~\ref{it:lem:lift-preserves-cons:own-plus-brr})
  The proof is by an easy calculation.
  Concretely, for each \( \adr \in \dom(H)\) and \( \lvarTwo \in \lenv \lift \lvar \), we have:
  \begin{align*}
    &\own^R_{\tyenv \lift \lvar}(\lvarTwo; \adr) \\
    &= \sum \{ \own(\ty) \mid x : \ty \in \tyenv \lift \lvar, \lft(\ty) = \lvarTwo, R(x) = \adr \} \tag{by def.\ of \( \own \)}\\
    &= \sum \{ \own(\ty' \lift \lvar) \mid x : \ty' \in \tyenv, \lft(\ty) = \lvarTwo,  R(x) = \adr \} \tag{by def.\ \( \tyenv \lift \lvar \) and \( \lvarTwo \neq \lvar \)}\\
    &=
      \begin{aligned}[t]
        & \sum \{ \own(\ty' \lift \lvar) \mid x : \ty' \in \tyenv, \lft(\ty') = \lvarTwo,  R(x) = \adr, \blft(\ty') = \alpha \} \\
        &+ \sum \{ \own(\ty' \lift \lvar) \mid x : \ty' \in \tyenv, \lft(\ty') = \lvarTwo,   R(x) = \adr, \blft(\ty') \neq \alpha \}
      \end{aligned} \tag{case analysis on \( \blft(\ty') \)}\\
    &=
      \begin{aligned}[t]
        & \sum \{ \own(\ty') + \bown(\ty') \mid x : \ty' \in \tyenv, \lft(\ty') = \lvarTwo,  R(x) = \adr, \blft(\ty') = \alpha\} \\
        &+ \sum \{ \own(\ty') \mid x : \ty' \in \tyenv,  \lft(\ty') = \lvarTwo, R(x) = \adr, \blft(\ty') \neq \alpha \}
      \end{aligned} \tag{by def.\ \( \tyenv \lift \lvar \)}\\
    &= \begin{aligned}
         &\sum \{ \own(\ty') \mid x : \ty' \in \tyenv, \lft(\ty') = \lvarTwo,  R(x) = \adr \}  \\
         &+ \sum \{ \bown(\ty') \mid x : \ty' \in \tyenv, \lft(\ty') = \lvarTwo, \blft(\ty') = \lvar, R(x) = \adr \}
       \end{aligned} \\
    &= \own_{\tyenv}^{R}(\lvarTwo; \adr) + \brr_\tyenv^R(\lvarTwo \bto \lvar; \adr).
  \end{align*}

  Equation~\eqref{eq:lem:lift-preserves-cons:own} follows by taking the sum over \( \lvarTwo \in \lenv \lift \lvar \) in the each side of Equation~\eqref{eq:lem:lift-preserves-cons:own-with-lft} and applying Lemma~\ref{lem:sum-own-and-brr}.

  \noindent(Proof of~\ref{it:lem:lift-preserves-cons:main})
  It is easy to show \( \lenv \lift \lvar \vdash_\mathrm{WF} \tyenv \lift \lvar \) from  \(\lenv \vdash_\mathrm{WF} \tyenv \),
  and the memory and type consistencies are trivial as we are not modifying the heap nor the register.
  Hence, we only check the fraction and borrow consistency.

  The fraction consistency holds by the following (in)equalities (which holds for every \( \adr \in \dom(H)\)):
  \begin{align*}
    \own_{\tyenv \lift \lvar}(\adr)
    &= \sum_{\beta \neq \alpha }\own_\tyenv^R(\lvarTwo; \adr) + \bby_\tyenv^R(\lvar; \adr) \tag{by~\ref{it:lem:lift-preserves-cons:own-plus-brr}} \\
    &\le  \sum_{\beta \neq \alpha }\own_{\Gamma }^{R}(\lvarTwo; \adr) + \own_\tyenv^R(\lvar; \adr) \tag{by Lemma~\ref{lem:borrow-cons-min}} \\
    &= \own_\tyenv^R(a) \tag{by Lemma~\ref{lem:sum-own-and-brr}} \\
    &\le 1 \tag{by the fraction consistency of \( \cons \lenv \tyenv H R \)}
  \end{align*}

  We now show the borrow consistency.
  Take \( \adr \in \dom(H) \) and \( \lvarTwo \in \lenv \lift \lvar \).
  We have
  \begin{align*}
    \bby^R_{\tyenv \lift \lvar}(\lvarTwo; \adr)
    &= \bby^R_{\tyenv}(\lvarTwo; \adr) \tag{by~\ref{it:lem:lift-preserves-cons:brr} and Lemma~\ref{lem:sum-own-and-brr}} \\
    &\le \own^R_\tyenv(\lvarTwo; \adr) + \bfrm^R_\tyenv(\lvarTwo; \adr) \tag{by the borrow consistency of \( \cons \lenv \tyenv H R\)} \\
    &= \own^R_{\tyenv \lift \lvar}(\lvarTwo; \adr) - \brr^R_\tyenv(\lvarTwo \bto \lvar; \adr) + \sum_{\lvarThree \in \lenv}\brr^R_\tyenv(\beta \bto \lvarThree; \adr) \tag{by~\ref{it:lem:lift-preserves-cons:own-plus-brr} and Lemma~\ref{lem:sum-own-and-brr}} \\
    &= \own^R_{\tyenv \lift \lvar}(\lvarTwo; \adr) + \sum_{\lvarThree \in \lenv \lift \lvar}\brr^R_\tyenv(\lvarTwo \bto \lvarThree; \adr) \\
    &= \own^R_{\tyenv \lift \lvar}(\lvarTwo; \adr) + \sum_{\lvarThree \in \lenv \lift \lvar}\brr^R_{\tyenv \lift \lvar}(\lvarTwo \bto \lvarThree; \adr) \tag{by~\ref{it:lem:lift-preserves-cons:brr}} \\
    &= \own^R_{\tyenv \lift \lvar}(\lvarTwo; \adr) + \bfrm^R_{\tyenv \lift \lvar}(\lvarTwo; \adr)  \tag{by Lemma~\ref{lem:sum-own-and-brr}} \\
  \end{align*}
  as desired.
  \qed
\end{proof}
We now prove that ownership transfers between aliasing references in \( \tyenv\) preserves \( \cons \lenv \tyenv H R \).
This is the lemma used to prove the preservation under reductions caused by \drvrule{RS-Let} and \drvrule{RS-Alias}.

\begin{lemma}
  \label{lem:cons-addition}
  Suppose that \( \cons \lenv \tyenv H R \).
  \begin{enumerate}
    \item If \( \tyenv = \tyenv_0, y : \ty \), \( \ty = \ty_x + \ty_y\) and \( x \notin \dom(R)\), then \( \cons \lenv {\tyenv_0, x : \ty_x, y : \ty_y} H {R \upd x {R(y)}}\)
      \label{it:lem:cons-addition:let}
    \item If \( \tyenv = \tyenv_0, x : \ty_x, y : \ty_y \), \( R(x) = R(y) \), \( \ty_x + \ty_y = \ty \) and \( \ty'_x + \ty'_y = \ty \) for some \( \ty \), then \( \cons \lenv {\tyenv_0, x : \ty'_x, y : \ty'_y} H R \) \label{it:lem:cons-addition:alias}
  \end{enumerate}
\end{lemma}
\begin{proof}
  \newcommand*{\tyenvTwo}{\tyenv'}
  We only prove~\ref{it:lem:cons-addition:let} because~\ref{it:lem:cons-addition:alias} can be proved in a similar manner.
  The case where \( \ty = \ty_x = \ty_y = \tint \) is trivial, so let us assume that \( \ty \) is a reference type.
  To simplify the notation, we write \( \adr \) for \( R(y) \), \( R' \) for \( R \upd x {R(y)}\)and \(  \tyenvTwo \) for \( \tyenv_0, x : \ty_x, y : \ty_y \).
  We first check the memory consistency.
  It suffices to show that \( R(x) \in \dom(H) \).
  Since \( \adr \in \dom(H) \) and \( H(\adr) \in \mathbb Z \) due to the memory consistency of  \( \cons \lenv \tyenv H R \), we have \( R(x) = \adr \in \dom(H) \) and \( H(R(x)) \in \mathbb Z \) as desired.
  To show the fraction and borrow consistencies, we proceed by a case analysis on the rule used to derive \( \ty = \ty_x + \ty_y \).

  Suppose that the rule used was \drvrule{A-Share}.
  Then it must be the case that \( \ty = \tref \lvar {r_1 + r_2}{\lvarTwo, s_1 + s_2} \), \( \ty_x =  \tref \lvar {r_1}{\lvarTwo, s_1} \) and \( \ty_y =  \tref \lvar {r_2}{\lvarTwo, s_2} \) for some fractions \( r_1 \), \( r_2 \), \( s_1 \), \( s_2\) and a lifetime variable \( \lvar \).
  We therefore have \( \own(\ty) = \own(\ty_x) + \own(\ty_y) \).
  From this it follows that \( \own_\tyenv^R(\adr) = \own^{R'}_{\tyenvTwo}(\adr) \); we also have \( \own_\tyenv^R(\adr') = \own^{R'}_{\tyenvTwo}(\adr') \) for any \( \adr'\) such that \( \adr' \neq \adr \) because \( x \) and \( y \) do not have any ownership for this address.
  Thus, the fraction consistency of \( \cons \lenv \tyenvTwo H {R'} \) follows from that of \( \cons \lenv \tyenv H R \).
  Similarly, by \( \own(\ty) = \own(\ty_x) + \own(\ty_y) \) and \( \bown(\ty) = \bown(\ty_x) + \bown(\ty_y) \), we have \( \own_\tyenv^R(\lvar; \adr) = \own^{R'}_{\tyenvTwo}(\lvar; \adr) \), \( \bfrm_\tyenv^R(\lvar; \adr) = \bfrm^{R'}_{\tyenvTwo}(\lvar; \adr) \) and \(  \bby_\tyenv^R(\lvarTwo; \adr) = \bby^{R'}_{\tyenvTwo}(\lvarTwo; \adr)\).
  Using these equations we get the borrow consistency of \( \cons \lenv \tyenvTwo H {R'} \) from that of \( \cons \lenv \tyenv H R \).

  Now we show the case for \drvrule{A-Borrow}.
  In this case, we have \( \ty = \tref \lvar {r + s} {} \), \( \ty_x =  \tref \lvarTwo s {} \) and \( \ty_y =  \tref \lvar r {\lvarTwo, s} \) for some fractions \( r \), \( s \) and a lifetime variables \( \lvar \) and \( \lvarTwo\); the case where \( x \) lends some ownership to \( y \) is symmetric and is proved similarly.
  Again, we have \( \own(\ty) = \own(\ty_x) + \own(\ty_y) \) and, from this, the fraction consistency follows as in the case for \drvrule{A-Share}.
  Observe that we have
  \begin{align*}
    \own^{R'}_{\tyenvTwo}(\lvar; \adr) + \bfrm^{R'}_{\tyenvTwo}(\lvar; \adr)
    &= (\own_\tyenv^R(\lvar; \adr) - s) + (\bfrm_\tyenv^R(\lvar; \adr)  + s) \\
    &= \own_\tyenv^R(\lvar; \adr) + \bfrm_\tyenv^R(\lvar; \adr) \\
    &\ge \bby^{R}_{\tyenv}(\lvar; \adr) \tag{by borrow consistency of \( \cons \lenv \tyenv H R \)} \\
    &= \bby^{R'}_{\tyenvTwo}(\lvar; \adr)
  \end{align*}
  and
  \begin{align*}
    \own^{R'}_{\tyenvTwo}(\lvarTwo; \adr) + \bfrm^{R'}_{\tyenvTwo}(\lvarTwo; \adr)
    &= (\own_\tyenv^R(\lvarTwo; \adr) + s) + \bfrm_\tyenv^R(\lvarTwo; \adr) \\
    &\ge \bby^{R}_{\tyenv}(\lvarTwo; \adr) + s \tag{by borrow consistency of \( \cons \lenv \tyenv H R \)} \\
    &= \bby^{R'}_{\tyenvTwo}(\lvarTwo; \adr).
  \end{align*}
  Since \( x \) and \( y \) does not have any ownership for address other than \( \adr \), the above two inequality is enough to conclude that the borrow consistency of \( \cons \lenv \tyenvTwo H {R'} \) holds.
  \qed
\end{proof}

\subsubsection{Lemmas about return contexts}
We prove two lemmas (Lemmas~\ref{lem:return-context-subst-var} and~\ref{lem:call-creates-return-context}) about return contexts.
Lemmas~\ref{lem:return-context-subst-var} and~\ref{lem:call-creates-return-context} are the key lemmas for the return, i.e.~\drvrule{RS-Var}, and call cases, respectively; we have separated them into separate lemmas for clarity.
\begin{lemma}
  \label{lem:return-context-subst-var}
  Suppose that
  \begin{gather*}
    \tyjudgement \fenv \lenv \tyenv {x : \ty} {\lenv'} {\tyenv'}\quad \text{and}\\
    \ejdg \fenv \ty {\lenv'} {\tyenv'} {\letin y \hole e} {\tyTwo} {\lenv''} {\tyenv''}.
  \end{gather*}
  Then \( \tyenv \) can be split as \( \tyenv = \tyenv_0 +  \genv \) such that
  \begin{gather*}
    \tyjudgement \fenv \lenv {\tyenv_0} {\letin y x  : \tyTwo} {\lenv''} {\tyenv''}.
  \end{gather*}
\end{lemma}
\begin{proof}
  Since  \( \ejdg \fenv \ty {\lenv'} {\tyenv'} {\letin y \hole e} {\tyTwo} {\lenv''} {\tyenv''} \), we must have
  \begin{equation}
    \tyjudgement \fenv {\lenv'} {\tyenv', y : \ty} {e : \tyTwo} {\lenv''} {\tyenv''} \label{eq:lem:return-context-subst-var:tyjudgement-e}.
  \end{equation}
  On the other hand, by inversion on \( \tyjudgement \fenv \lenv \tyenv {x : \ty} {\lenv'} {\tyenv'} \), it must be the case that
  \begin{gather*}
    \tyenv = \genv + \tyenv' + x : \ty \quad \text{and} \quad \lenv' = \lenv
  \end{gather*}
  for some \( \genv \).
  Without loss of generality, assume that \( x \in \dom(\tyenv') \), i.e.~\( \tyenv' = \tyenv'_0, x : \tau_x' \) for some \( \tau_x' \); otherwise we can extend \( \tyenv' \) with \( x \) using the weakening lemma (Lemma~\ref{lem:weakening}).
  We take \( \tyenv' + x : \tau \) as \( \tyenv_0 \).
  Since \( \lenv' = \lenv \) and  \( \tyenv' = \tyenv'_0, x : \tau_x' \), the judgment~\eqref{eq:lem:return-context-subst-var:tyjudgement-e} is of the form
  \begin{equation*}
    \tyjudgement \fenv {\lenv} {\tyenv'_0, x : \ty_x',  y : \ty} {e : \tyTwo} {\lenv''} {\tyenv''},
  \end{equation*}
  and, since \( \tyenv_0(x) = \ty_x' + \ty \), applying \drvrule{T-Let} yields
  \begin{equation*}
    \tyjudgement \fenv {\lenv} {\tyenv' +  x : \ty} {\letin y x e : \tyTwo} {\lenv''} {\tyenv''}.
  \end{equation*}
  \qed
\end{proof}

\begin{lemma}
  \label{lem:call-creates-return-context}
  Suppose that
  \begin{gather}
    \tyjudgement \fenv \lenv \tyenv {\letin{x}{f\langle\vv{\beta}\rangle(y_1, \dots, y_n)} e'  : \tyTwo} {\lenv'} {\tyenv'}  \label{eq:lem:call-creates-return-context:call-typing} \\
    \vdash_\mathrm{WF} \fenv \qquad \fenv \vdash f \mapsto \tuple{\vv{\alpha}}(x_1, \dots, x_n) e  \label{eq:lem:call-creates-return-context:fundef-typing} \\
    \fenv(f) = \forall \vv \lvar : \mathcal M.\ \langle \tau_1, \dots, \tau_n \rangle \to \langle \tau_1', \dots, \tau_n' \mid \tau \rangle \label{eq:lem:call-creates-return-context:f-type}\\
    f \mapsto \tuple{\vv{\alpha}}(x_1, \dots, x_n) e \in D
  \end{gather}
  Then we have
  \begin{gather}
    \tyjudgement \fenv \lenv \tyenv {\sub {\seq y} {\seq x} {e} :  {\sub {\seq \lvarTwo} {\seq \lvar} \ty}} {\lenv} {\tyenv''} \label{eq:lem:call-creates-return-context:goal1} \\
    \ejdg \fenv {\sub {\seq \lvarTwo} {\seq \lvar} \ty} \lenv {\tyenv''} {\letin x \hole e'} \tyTwo {\lenv'} {\tyenv'} \label{eq:lem:call-creates-return-context:goal2}
  \end{gather}
  where
  \begin{align*}
     \tyenv'' &\defrel \sub {y_1: \tyTwo_1'}{x_1:\tyTwo_1}\ \cdots \sub {y_n: \tyTwo_n'}{x_n:\tyTwo_n}\tyenv \\
    \tyTwo_i &\defrel  \sub {\seq \lvarTwo} {\seq \lvar} \ty_i  \qquad \tyTwo'_i \defrel  \sub {\seq \lvarTwo} {\seq \lvar} \ty'_i
  \end{align*}
\end{lemma}
\begin{proof}
  Note that by inversion on~\eqref{eq:lem:call-creates-return-context:call-typing}, we have
  \begin{align}
    \tyjud{\Gamma_0, x:\tyTwo, y_1:\tyTwo_1', \dots, y_n:\tyTwo_n'}{e : \tyTwo}{\Gamma'} \label{eq:lem:call-creates-return-context:inversion-call1}\\
    x \notin \dom(\tyenv') \qquad  \tyenv = \tyenv_0, x:\tyTwo, y_1:\tyTwo_1, \dots, y_n:\tyTwo_n \label{eq:lem:call-creates-return-context:inversion-call2}
  \end{align}

  We first show that \eqref{eq:lem:call-creates-return-context:goal2} holds.
  By the definition of \drvrule{TC-Let}, it suffices to show that
  \begin{align*}
    \tyjudgement \fenv \lenv {\tyenv'', x : \ty} {e : \tyTwo} {\lenv'} {\tyenv'} \quad \text{and} \quad x \notin \dom(\tyenv')
  \end{align*}
  These are exactly what~\eqref{eq:lem:call-creates-return-context:inversion-call1} and~\eqref{eq:lem:call-creates-return-context:inversion-call2} assert.

  Now we show that~\eqref{eq:lem:call-creates-return-context:goal1} holds.
  By inversion on~\eqref{eq:lem:call-creates-return-context:fundef-typing} together with~\eqref{eq:lem:call-creates-return-context:f-type}, we have
  \begin{align*}
    \tyjudgement \fenv {\mathcal{M}} {x_1:\tau_1, \dots, x_n:\tau_n}{e : \tau}{\mathcal{M}}{x_1:\tau_1', \dots, x_n:\tau_n'}
  \end{align*}
  Thanks to substitution lemmas (Lemmas~\ref{lem:substitution} and~\ref{lem:substitution-lvar}), from the above judgment, we obtain
  \begin{align*}
    \tyjudgement \fenv {\sub {\seq \lvarTwo} {\seq \lvar} {\mathcal M}} {y_1:\tyTwo_1, \dots, y_n:\tyTwo_n}{\sub {\seq y} {\seq x} e :  \sub {\seq \lvarTwo} {\seq \lvar} \tau}{ \sub {\seq \lvarTwo} {\seq \lvar} {\mathcal M}}{y_1:\tyTwo_1', \dots, y_n:\tyTwo_n'}
  \end{align*}
  Then, by the weakening lemmas (Lemmas~\ref{lem:weakening2} and~\ref{lem:weakening-lvar}), we have
  \begin{align*}
    \tyjudgement \fenv {\mathcal L} \tyenv {\sub {\seq y} {\seq x} e :  \sub {\seq \lvarTwo} {\seq \lvar} \tau}{\mathcal L}{\tyenv''}
  \end{align*}
    \qed
\end{proof}

\subsection{Main Proof of Type Preservation}
Finally, we prove that the one-step reduction relation preserves well-typedness (of configurations).
By repeatedly applying this result, we obtain \cref{thm:ownsum-bound} because \( \vdash_D \langle \emptyset \emptyset, \cdot, e \rangle\) holds for any program \( \langle D, e \rangle\) such that \( \vdash \langle D, e \rangle\)
\begin{lemma}
  \label{lem:subject-reduction}
  If \( \vdash  \configsrc H R {\vv F} e \) and \( \configsrc H R {\vv F} e \reduce D  \configsrc {H'} {R'} {\vv {F'}} {e'} \), then \( \vdash  \configsrc {H'} {R'} {\vv {F'}} {e'} \).
\end{lemma}
\begin{proof}
  \newcommand*{\tyenvTwo}{\tyenv''}
  \newcommand*{\lenvTwo}{\lenv''}
  \newcommand*{\totenvTwo}{\totenv'}
  \newcommand*{\genvTwo}{\Delta'}

  The proof is by a case analysis on the reduction rules.

  Before we proceed to each case, we explain the overall structure of the proof.
  Our goal is to find find \( \lenvTwo \), \( \totenvTwo \), \( \tyenvTwo \) such that
  \begin{gather*}
    \cons \lenvTwo \totenvTwo {H'} {R'} \qquad
    \tyjudgement \fenv {\lenvTwo} {\tyenvTwo} {e': \tyTwo} {\lenv'} {\tyenv'}
  \end{gather*}
  such that \( \totenvTwo = \tyenvTwo + \genvTwo \) for some \( \genvTwo \), and also to show that stack \( \vv{F'} \) is well-typed by using \( \vdash  \configsrc H R {\vv F} e \).
  (In case of \drvrule{Rs-Var} and \drvrule{Rs-Call}, we also need to find suitable \( \ty \), \( \lenv' \) and \( \tyenv' \) as the function we are reducing at that moment changes.)

  From  \( \vdash  \configsrc H R {\vv F} e \), as assumptions, we have
  \begin{gather}
    \cons \lenv \totenv H R  \label{eq:lem:subject-reduction:cons}\\
    \tyjudgement \fenv {\lenv} {\tyenv} {e: \tyTwo} {\lenv'} {\tyenv'} \label{eq:lem:subject-reduction:e-typing}
  \end{gather}
  where \( \totenv = \tyenv + \genv \).
  (Throughout the proof we will fix \( \fenv \), \( \tyenv \) and \( \totenv \) to be the environments satisfying the above relations.)
  In most cases, \( \lenvTwo\), \( \tyenvTwo \), \( \totenvTwo \) are the same as \( \lenv \), \( \tyenv \), \( \totenv \), respectively or simple modifications of \( \lenv \), \( \tyenv \), \( \totenv \).

  We first consider the two cases that involves some reasoning on stacks.
    \caset{Case: \sfdrvrule{Rs-Var}}
      In this case, we must have
      \begin{gather*}
        H' = H, \quad R' = R, \quad \vv F = F_n : \cdots : F_1, \quad  \vv {F'} = F_{n - 1} : \cdots : F_1 \\
        e = x, \quad e' = F_n[x]
      \end{gather*}
      for some \( n \ge 1 \).
      Moreover, by the definition of \( \vdash  \configsrc H R {\vv F} e \), we must also have
      \begin{gather*}
        \ejdg \fenv {\ty_i} {\lenv_i} {\tyenv_i} {F_i} {\ty_{i-1}} {\lenv_{i-1}} {\tyenv_{i-1}} \qquad ( 1 \le i \le n ) \\
        \tyjudgement \fenv \lenv \tyenv {x : \ty_n} {\lenv_n} {\tyenv_n}.
      \end{gather*}
      where \( \ty_n = \tyTwo \), \( \lenv_n = \lenv' \) and \( \tyenv_n = \tyenv' \).
      Therefore, by Lemma~\ref{lem:return-context-subst-var}, we can split \( \tyenv \) as \( \tyenv_0 + \genv' \) satisfying
      \[
      \tyjudgement \fenv \lenv {\tyenv_0} {F_n[x] : \ty_{n - 1}} {\lenv_{n - 1}} {\tyenv_{n - 1}}.
      \]
      Hence, we can take \( \tyenv_0 \) as \( \tyenvTwo \), \( \lenv \) as \( \lenvTwo \) and \( \totenv \) as \( \totenvTwo\).
      Note that \( \totenv = \tyenv + \genv = \tyenv_0 + (\genv' + \genv) \) for some \( \genv \).
      The relation \( \cons \lenvTwo \totenvTwo {H'} {R'} \) trivially holds because we have \( \cons \lenv \totenv H R \).
    \caset{Case: \sfdrvrule{Rs-Call}}
      In this case, we must have
      \begin{gather*}
        H' = H, \quad R' = R, \\ %
        e = \letin{x}{f\langle\vv{\beta}\rangle(y_1, \dots, y_n)}e_0, \quad e' = \sub {\vv y} {\vv x} e_f \\
        \vv{F}' = (\letin {x'} {\hole} \sub{x'}{x} e_0) : \vv F
       \end{gather*}
      for some \( f \mapsto \tuple{\vv{\alpha}}(x_1, \ldots, x_n) e_f \in D\).
      Let \( F_n \) be the return context on the top of the stack \( \vv F \).
      Note that \( F_n \) needs to be typed as follows:
      \begin{align*}
        \ejdg \fenv \tyTwo {\lenv'} {\tyenv'} {F_n} {\ty_{n-1}} {\lenv_{n-1}} {\ty_{n - 1}}
      \end{align*}
      for some \( \ty_{n-1} \),  \( \lenv_{n-1} \) and  \( \ty_{n - 1} \).
      Since \( \fenv \vdash D \), we also have
      \begin{align}
        \vdash_{WF} \fenv \qquad \fenv \vdash f \mapsto \tuple{\vv{\alpha}}(x_1, \ldots, x_n) e_f \label{eq:lem:subject-reduction:call:fundef-typing}
      \end{align}

      For \( \lenvTwo \), \( \tyenvTwo \) and \( \totenv \),  we take \( \lenv \), \( \tyenv \) and \( \totenv \), respectively.

      Our task is to show that
      \begin{gather*}
        \ejdg \fenv {\ty_{n + 1}} \lenv {\tyenv_{n + 1}} {\letin {x'} {\hole} \sub{x'}{x} e_0} \tyTwo \lenv \tyenv \\
        \tyjudgement \fenv \lenv \tyenv {\sub {\vv y} {\vv x} e_f : \ty_{n + 1}} \lenv {\tyenv_{n + 1}}
      \end{gather*}
      for some \( \ty_{n + 1} \) and \( \tyenv_{n + 1}\).
      These type and type environment can be obtained by applying Lemma~\ref{lem:call-creates-return-context} to~\eqref{eq:lem:subject-reduction:call:fundef-typing} and to the type judgment of \( e \), i.e.~\eqref{eq:lem:subject-reduction:e-typing} (with the help of substitution lemma, i.e.~Lemma~\ref{lem:substitution}, to substitute \( x ' \) to \( x \)).
      Concretely,
      \begin{align*}
        \ty_{n + 1 } &= \sub {\vv \lvarTwo} {\vv \lvar} \ty_f \quad \text{and} \\
        \tyenv_{n + 1}& = \sub {y_1 : {\sub {\vv \lvarTwo} {\vv \lvar} {\ty_1'}}}{x_1 : \ty_1}  \cdots \sub {y_n : {\sub {\vv \lvarTwo} {\vv \lvar} {\ty_n'}}}{x_n : \ty_n} \tyenv
      \end{align*}
      provided that  \( \fenv(f) = \forall \vv \lvar : \mathcal M.\ \langle \tau_1, \dots, \tau_n \rangle \to \langle \tau_1', \dots, \tau_n' \mid \ty_f \rangle\).
  \par\bigskip
  So far we are done with the case \drvrule{Rs-Var} and \drvrule{Rs-Call}.
  The remaining cases do not involve function calls or returns.
  Therefore, in what follows, we will omit the reasoning on call stacks.
  \caset{Case: \sfdrvrule{Rs-Let}}
      It must be the case that
      \begin{gather}
        H' = H, \nonumber \\
        R' = R \upd {x'} {R(y)}, \quad x' \notin \dom(R) \label{eq:lem:subject-reduction:let:R-prime}\\
        e = \letin x y  e_0, \quad \sub {x'} x {e_0} = e' \nonumber \\
        \tyenv = \tyenv_0, y : \ty
      \end{gather}
      for some \( x \), \( y \), \( \ty \), \(e_0\) and \( \tyenv_0 \).

      By inversion on the typing rule, we obtain
      \begin{gather}
        \tyjudgement \fenv {\lenv} {\tyenv_0, y : \ty_y, x : \ty_x} {e_0 : \tyTwo} {\lenv'} {\tyenv'} \nonumber \\
        \ty = \ty_x + \ty_y \label{eq:lem:subject-reduction:let:addition} \\
        x \notin \dom(\tyenv') \label{eq:lem:subject-reduction:let:not-in-tyenv}
      \end{gather}
      By the substitution lemma (Lemma~\ref{lem:substitution}) and~\eqref{eq:lem:subject-reduction:let:not-in-tyenv}, we also have
      \begin{gather*}
        \tyjudgement \fenv {\lenv} {\tyenv_0, y : \ty_y, x' : \ty_x} {e' : \tyTwo} {\lenv'} {\tyenv'}.
      \end{gather*}
      We set \( \lenvTwo \defrel \lenv \) and \( \tyenvTwo = \tyenv_0, y : \ty_y, x' : \ty_x \).
      Furthermore, we define \( \totenvTwo \defrel  (\tyenv_0, y : \tyTwo_y, x' : \tyTwo_x) + \genv  \).
      The remaining goal is to show \( \cons \lenv \totenvTwo H R\), which can be shown by \ref{it:lem:cons-addition:let} of~Lemma~\ref{lem:cons-addition} together with \eqref{eq:lem:subject-reduction:let:R-prime}, \eqref{eq:lem:subject-reduction:let:addition} and \( \totenv = (\tyenv, y : \ty) + \genv \).
      (Note that \( x' \notin \dom(\genv)\) because of type consistency.)
    \caset{Case: \sfdrvrule{Rs-Arith}}
      Similar (but also easier compared) to the previous case.
    \caset{Case: \sfdrvrule{Rs-MkRef}}
      In this case we have,
      \begin{gather*}
        H' = H \upd \adr {R(y)}, \quad \adr \notin \dom(H), \qquad R'= R \upd {x'} \adr, \quad x' \notin \dom(R) \\
        e = \letin x {\mkref y} e_0, \quad \sub {x'} x {e_0} = e' \\
        \tyenv = \tyenv_0, y : \tint
      \end{gather*}
      for some \( \adr \), \( x \), \( x' \) and \( y \).
      By inversion on the typing rule, we must also have
      \begin{gather}
        \tyjudgement \fenv \lenv {\tyenv_0, x : \tref \lvar 1 {}, y : \tint} {e_0 : \tyTwo} {\lenv'} {\tyenv'} \nonumber \\
        x \notin \dom(\tyenv') \label{eq:lem:subject-reduction:mkref:not-in-tyenv}
      \end{gather}
      Again, by applying the substitution lemma (Lemma~\ref{lem:substitution}) together with~\eqref{eq:lem:subject-reduction:mkref:not-in-tyenv} to this judgment gives
      \begin{gather*}
        \tyjudgement \fenv \lenv {\tyenv_0, x' : \tref \lvar 1 {}, y : \tint} {e' : \tyTwo} {\lenv'} {\tyenv'}
      \end{gather*}
      We take \( \tyenv_0, x' : \tref \lvar 1 {}, y : \tint \) as \( \tyenvTwo \), \( \lenv \) as \( \lenvTwo \) and \( \totenv, x' : \tref \lvar 1 {} \) as \( \totenvTwo \).

      Our goal is to show that \( \cons \lenv {\totenv, x' : \tref \lvar 1 {}} {H'} {R'} \) under the assumption that \( \cons \lenv {\totenv} H R \).
      To check the memory consistency, we need to check that the newly introduced variable \( x' \) does not cause any out-of-bound access.
      Since \( R'(x') = \adr \), we have \( R'(x') \in \dom (H') \) as desired.
      We also have \( H'(R(x')) = R(y) \in \mathbb Z \) from the type consistency of \( y \).
      Type consistency for \( x' \) is also clearly met.

      Now we verify the fraction and borrow consistency.
      Observe that,
      \begin{align}
        \forall z \in \dom(\totenv) .\ R(z) \neq \adr  \label{eq:lem:subject-reduction:mkref:adr-does-not-exist}
      \end{align}
      because of the memory consistency of \( \cons \lenv \totenv H R \) together with \( \adr \notin \dom(H) \).
      Therefore,
      \begin{align*}
        \own_{\totenvTwo}^{R'}(\adr) = \own(\tref \lvar 1 {}) = 1
      \end{align*}
      For \( \adr' \neq \adr \), it is easy to check that \( \own_{\totenvTwo}^{R'}(\adr') = \own_{\totenv}^R(\adr') \), which does not exceed \( 1 \) because of the fraction consistency of \( \cons \lenv \totenv H R \).
      Hence, the fraction consistency holds.
      To check the borrow consistency, it suffices to check
      \begin{align}
      \bby^{R'}_{\totenvTwo}(\lvar; \adr) \le \own^{R'}_{\totenvTwo}(\lvar; \adr) + \bfrm_{\totenvTwo}^{R'}(\lvar; \adr) \label{eq:lem:subject-reduction:mkref:brr-goal}
      \end{align}
      because the other cases follows from the borrow consistency of \( \cons \lenv {\totenv} H R \).
      By~\eqref{eq:lem:subject-reduction:mkref:adr-does-not-exist}, we have
      \[
      \bby^{R'}_{\totenvTwo}(\lvar; \adr) =  \bfrm_{\totenvTwo}^{R'}(\lvar; \adr) = 0.
      \]
      Hence,~\eqref{eq:lem:subject-reduction:mkref:brr-goal} holds because the left-hand side is \( 0 \) whereas the right-hand side is \( 1 \).
  \caset{Case: \sfdrvrule{Rs-Alias}}
      It must be the case that
      \begin{gather*}
        H'= H, R' = R, \vv{F}' = \vv F, e = \alias {x = y}  {e'}
      \end{gather*}
      with \( R(x) = R(y) = \adr \) for some \( \adr \).
      Moreover, the type environment \( \tyenv \) must be of the form \( \tyenv_0, x : \ty_x, y : \ty_y \) for some \( \ty_ x \) and \( \ty_y \).
      By inversion on the typing of \( e \), that is~\eqref{eq:lem:subject-reduction:e-typing}, we have
      \begin{gather*}
        \tyjudgement \fenv {\lenv} {\tyenv_0, y : \tyTwo_y, x : \tyTwo_x} {e' : \tyTwo} {\lenv'} {\tyenv'}.
      \end{gather*}
      such that \( \ty = \ty_x + \ty_y \) and \( \ty = \tyTwo_x + \tyTwo_y \).
      We take \( \lenv \), \( \tyenv_0, y : \tyTwo_y, x : \tyTwo_x \) \( (\tyenv_0, y : \tyTwo_y, x : \tyTwo_x) + \genv \) for \( \lenvTwo \), \( \tyenvTwo \) and \( \totenvTwo \), respectively.
      Note that we have \( \totenv(x) + \totenv(y) = \ty' = \totenvTwo(x) + \totenvTwo(y) \) for some \( \ty' \) because \( \totenv = (\tyenv_0, x : \ty_x, y : \ty_y ) + \genv \). %
      From this and \( R'(x) = R'(y) \), we can use~\ref{it:lem:cons-addition:alias} of Lemma~\ref{lem:cons-addition} to conclude \( \cons \lenvTwo \totenvTwo {H'} {R'} \).
  \caset{Case: \sfdrvrule{Rs-Newlft}}
      In this case, we must have
      \begin{gather*}
        H'= H, R' = R, e = \newlftin \lvar {e'}.
      \end{gather*}
      By inversion on the typing rule, we have
      \begin{gather*}
        \tyjudgement{\Theta}{\mathcal{L}\cup\{\alpha\lftlt \beta \mid \beta \in \setlft(\mathcal{L})\}}{\Gamma}{e : \rho}{\mathcal{L}'}{\Gamma'}.
      \end{gather*}
      We set \( \lenvTwo \defrel  \lenv \cup\{\alpha\lftlt \beta \mid \beta \in \lenv )\} \), \( \tyenvTwo \defrel \tyenv \) and \( \totenvTwo \defrel \totenv \).
      Because \( \lenv \vdash_\mathrm{WF} \tyenv \) and \( \lenv \vdash_\mathrm{WF} \totenv\)  it follows that \( \lenvTwo \vdash_\mathrm{WF} \tyenv \) and \( \lenvTwo \vdash_\mathrm{WF} \totenv \).
      The well-formedness also ensures that \( \lvar \) does not appear in \( \tyenv \) nor \( \totenv \).
      This implies that
      \[
      \bby_{\totenv}^R(\lvar, a) =  \own_{\totenv}^R(\lvar, a) = \bfrm_{\totenv}^R(\lvar, a) = 0
      \]
      for every address \( \adr \).
      Therefore, \( \cons \lenvTwo \totenvTwo {H'} {R'} \) immediately follows from \( \cons \lenv \totenv H R \).
    \caset{Case: \sfdrvrule{Rs-Endlft}}
      In this case, we have
      \begin{gather*}
        H'= H, \quad  R' = R, \quad  e = \lftend \lvar; {e'}.
      \end{gather*}
      We must also have
      \begin{gather*}
        \tyjudgement \fenv {\lenv \lift \lvar} {\tyenv \lift \lvar} {e' : \tyTwo} {\lenv'} {\tyenv'} \quad \lvar \in \min(\lenv)
      \end{gather*}
      by inversion on the typing rule.
      We set \( \lenvTwo \defrel \lenv \lift \lvar \), \( \tyenvTwo \defrel \tyenv \lift \lvar \) and \( \totenvTwo = \totenv \lift \lvar\).
      It remains to show that \( \cons {\lenv \lift \lvar} {\totenv \lift \lvar} H R \), and this is exactly what Lemma~\ref{lem:lift-preserves-cons} states.

    \caset{Case: \sfdrvrule{Rs-Deref}}
      It must be the case that
       \begin{gather*}
         H'= H, \quad  R' = R\{x' \mapsto v\}, \quad H(R(y)) = v, \quad  x' \notin \dom(R) \\
         e = \letin{x}{\star y}e_0 \quad e' = \sub {x'} x e_0
      \end{gather*}
      for some \( x \), \( y \) and \( x' \).
      By inversion on~\ref{eq:lem:subject-reduction:e-typing}, we have
      \begin{gather*}
        \tyjudgement \fenv \lenv {\tyenv_0, x : \tint, y : \tref \lvar r B} {e  : \tyTwo} {\lenv'} {\tyenv'} \\
       x \notin \dom(\tyenv') \quad \tyenv = \tyenv_0, y : \tref \lvar r B
      \end{gather*}
      for some reference type \( y : \tref \lvar r B \) and type environment \( \tyenv_0 \).
      Applying the substitution lemma (Lemma~\ref{lem:substitution}) to the above judgment yields
      \begin{align*}
        \tyjudgement \fenv \lenv {\tyenv_0, x' : \tint, y : \tref \lvar r B} {\sub {x'} x e  : \tyTwo} {\lenv'} {\tyenv'}
      \end{align*}
      because \( x' \notin \dom(\tyenv')\).
      We define \( \lenvTwo \defrel \lenv \), \( \tyenvTwo \defrel \tyenv_0, x' : \tint, y : \tref \lvar r B\) and \( \totenvTwo = (\tyenv_0, x' : \tint, y : \tref \lvar r B) + \genv \).

      We are left to show that \( \cons \lenv \totenvTwo H {R\{x' \mapsto v\}} \).
      Fraction, borrow, memory consistency obviously hold since the only difference between \( \totenv \) and \(  \totenvTwo \) is the additional \emph{integer} variable \( x' \).
      The type consistency also holds because \( R'(x') = H(R(y)) \in \mathbb Z \) from the memory consistency of \eqref{eq:lem:subject-reduction:cons}.
    \caset{Case: \sfdrvrule{Rs-Assign}}
      In this case, we have
      \begin{gather*}
        H'= H\{\adr \mapsto R(y)\}, \quad  R' = R, \quad R(x) = \adr  \\
         e = x := y; e'
      \end{gather*}
      Furthermore, we have
      \begin{gather*}
        \tyjudgement \fenv \lenv \tyenv {e' : \tyTwo} {\lenv'} {\tyenv'}
      \end{gather*}
      by inversion on~\eqref{eq:lem:subject-reduction:e-typing} (the typing of \( e \)).
      Therefore, we can take \( \lenv \), \( \tyenv \) and  \( \totenv \) for \( \lenvTwo \), \( \tyenvTwo \) and \( \totenv \), respectively.

      It remains to show that \( \cons \lenv \totenv {H'} R \).
      Since the type environment hasn't changed from \( \totenv \), the fraction, borrow and type consistency trivially hold.
      We need to check the memory consistency since the value stored in \( \adr \) has changed.
      We have \( H(R(x)) = R(y) \in \mathbb Z \) because of the type consistency of \( y \).

    \caset{Case: \sfdrvrule{Rs-IfTrue}}
      In this case, we have
      \begin{gather*}
        H' = H, \quad R' = R \\
        e = \ifz x \then e' \els e_2
        \end{gather*}
      for some \( x \) and \( e_2\).
      By inversion on~\eqref{eq:lem:subject-reduction:e-typing} (the typing of \( e \)) we have
      \begin{align*}
        \tyjudgement \fenv {\lenv} {\tyenv} {e': \tyTwo} {\lenv'} {\tyenv'}.
      \end{align*}
      Hence, we can simply take \( \lenv \), \( \tyenv \) and  \( \totenv \) for \( \lenvTwo \), \( \tyenvTwo \) and \( \totenv \), respectively.
    \caset{Case: \sfdrvrule{Rs-IfFalse}}
      Almost identical to the previous case.
    \caset{Cases: \sfdrvrule{Rs-AliasFail} and \sfdrvrule{Rs-Fail}}
      These cases contradict with the assumption.
  \qed
\end{proof}

\section{Definitions Omitted from \cref{section:translation}}

\subsection{Operational Semantics for the Target Language}
\label{sec:operational-semantics-target}
Here, we give the complete definition of the reduction semantics of the target language that has been omitted from~\cref{sec:target-lang}.
The rules are given in~\cref{fig:target-opsem-comp}.
We write \(S\{x \mapsto v\}\) (where \(x \in \setvar\) and \(v \in \tgtval\)) for a register that maps the variable \( x \) to \(v\) and inherits \( S \) for the other variables.
We write \(\textbf{Refresh}_{\toml{S}}(t)\) for the term taken by renaming all the bound variables in \(t\) to fresh names (not contained in \(\dom(S)\) for any \(S \in \toml{S}\)).
This is for avoiding name conflicts in recursive function calls.
The partial functions \(\pi_1, \pi_2: \tgtval \to \tgtval\) are projections for pairs.
\begin{figure}
    \begin{center}
    \AxiomC{}
    \RightLabel{\smaller (\textsc{Rt-Var})}
    \UnaryInfC{\(\configtgt{\toml{D}}{\toml{S}}{(\tletin{x'}{[]} t) : \vv{\toml{F}}}{x} \longrightarrow \contgt{[x/x']t}\)}
    \DisplayProof
    \end{center}
    \begin{center}
        \AxiomC{\(\forall S\in\toml{S}.\ \llbracket o \rrbracket_S = n \)}
        \AxiomC{\(\toml{S}' = \{S\{x \mapsto n\} \mid S \in \toml{S}\}\)}
    \RightLabel{(\smaller \textsc{Rt-LetArith})}
    \BinaryInfC{\(\contgt{\tletin{x}{o}t} \longrightarrow \configtgt{\toml{D}}{\toml{S}'}{\vv{\toml{F}}}{t}\)}
    \DisplayProof
    \end{center}
    \begin{center}
        \AxiomC{\(\toml{S}' = \{S\{x \mapsto S(y)\} \mid S \in \toml{S}\}\)}
    \RightLabel{\smaller (\textsc{Rt-Let})}
    \UnaryInfC{\(\contgt{\tletin{x}{y}t} \longrightarrow \configtgt{\toml{D}}{\toml{S}'}{\vv{\toml{F}}}{t}\)}
    \DisplayProof
    \end{center}
    \begin{center}
        \AxiomC{\(f \mapsto (x_1, \dots, x_n) t \in D\)}
        \AxiomC{\(t'' = \textbf{Refresh}_{\toml{S}}(t)\)}
    \RightLabel{\smaller (\textsc{Rt-Call})}
    \BinaryInfC{\(\begin{array}{l}
        \contgt{\tletin{x}{f(y_1, \dots, y_n)}t'} \\
        \quad \longrightarrow \configtgt{\toml{D}}{\toml{S}}{(\tletin{x}{[]}t'):\vv{\toml{F}}}{[y_1/x_1]\cdots[y_n/x_n]t''}
    \end{array}\)}
    \DisplayProof
    \end{center}
    \begin{center}
        \AxiomC{\(\toml{S}' = \{S\{x' \mapsto \tuple{S(y), S(z)}\} \mid S \in \toml{S}\}\)}
    \RightLabel{\smaller (\textsc{Rt-LetPair})}
    \UnaryInfC{\(\contgt{\tletin{x}{\langle y, z \rangle}t} \longrightarrow \configtgt{\toml{D}}{\toml{S}'}{\vv{\toml{F}}}{t}\)}
    \DisplayProof
    \end{center}
    \begin{center}
        \AxiomC{\(\toml{S}' = \{S\{x' \mapsto \pi_1(S(y))\} \mid S \in \toml{S}\}\)}
    \RightLabel{\smaller (\textsc{Rt-LetFst})}
    \UnaryInfC{\(\contgt{\tletin{x}{\fst y}t} \longrightarrow \configtgt{\toml{D}}{\toml{S}'}{\vv{\toml{F}}}{t}\)}
    \DisplayProof
    \end{center}
    \begin{center}
        \AxiomC{\(\toml{S}' = \{S\{x' \mapsto \pi_2(S(y))\} \mid S \in \toml{S}\}\)}
    \RightLabel{\smaller (\textsc{Rt-LetSnd})}
    \UnaryInfC{\(\contgt{\tletin{x}{\snd y}t} \longrightarrow \configtgt{\toml{D}}{\toml{S}'}{\vv{\toml{F}}}{t}\)}
    \DisplayProof
    \end{center}
    \begin{center}
        \AxiomC{\(\toml{S}' = \{S\{x \mapsto n\} \mid S \in \toml{S}, n \in \mathbb{Z}\}\)}
    \RightLabel{\smaller (\textsc{Rt-LetNondet})}
    \UnaryInfC{\(\contgt{\tletin{x}{\star}t} \longrightarrow \configtgt{\toml{D}}{\toml{S}'}{\vv{\toml{F}}}{t}\)}
    \DisplayProof
    \end{center}
    \begin{center}
    \AxiomC{\(\toml{S}' = \{S \in \toml{S} \mid S(x) = S(y)\}\)}
    \RightLabel{\smaller (\textsc{Rt-Assume})}
    \UnaryInfC{\(\contgt{\tassume{x = y}; t} \longrightarrow \configtgt{\toml{D}}{\toml{S}'}{\vv{\toml{F}}}{t}\)}
    \DisplayProof
    \end{center}
    \begin{center}
        \AxiomC{\(\toml{S}' = \{ S \in \toml{S} \mid S(x) = 0\}\)}
        \AxiomC{\(\toml{S}' \neq \emptyset\)}
    \RightLabel{\smaller (\textsc{Rt-IfTrue})}
    \BinaryInfC{\(\contgt{\tifz x \tthen t_1 \tels t_2} \longrightarrow \configtgt{\toml{D}}{\toml{S}'}{\vv{\toml{F}}}{t_1}\)}
    \DisplayProof
    \end{center}
    \begin{center}
        \AxiomC{\(\toml{S}' = \{ S \in \toml{S} \mid S(x) \neq 0\}\)}
        \AxiomC{\(\toml{S}' \neq \emptyset\)}
    \RightLabel{\smaller (\textsc{Rt-IfFalse})}
    \BinaryInfC{\(\contgt{\tifz x \tthen t_1 \tels t_2} \longrightarrow \configtgt{\toml{D}}{\toml{S}'}{\vv{\toml{F}}}{t_2}\)}
    \DisplayProof
    \end{center}
    \begin{center}
    \AxiomC{}
    \RightLabel{\smaller (\textsc{Rt-Fail})}
    \UnaryInfC{\(\contgt{\tfail} \longrightarrow \textbf{Fail}\)}
    \DisplayProof
    \end{center}
    \caption{Operational semantics of the target language}
    \label{fig:target-opsem-comp}
\end{figure}

\newpage
\subsection{Translation Rules}
\label{sec:complete-translation}
\begin{figure}
\begin{center}
    \AxiomC{\(\trans{\Gamma, x:\tau_x, y:\tau_y}{e \COL \rho}{t}\)}
    \AxiomC{\( \tau = \tau_x + \tau_y \)}
    \AxiomC{\(t' = \convalias{\textbf{Nullify}(\tau) \shortrightarrow \rho_x }{\tau \shortrightarrow \rho_y}(x, y, t)\)}
\RightLabel{\smaller (\textsc{C-Let})}
\TrinaryInfC{\(\trans{\Gamma, y:\tau}{\letin{x}{y} e \COL \rho}{\tletin{x}{\tuple{\_, \_}} t'}\)}
\DisplayProof
\end{center}
\begin{center}
    \AxiomC{\(\trans{\Gamma, x : \tint}{e : \rho}{t}\)}
\RightLabel{\smaller (\textsc{C-Arith})}
\UnaryInfC{\(\trans{\Gamma}{\letin{x}{o} e : \rho}{\tletin{x}{o} t}\)}
\DisplayProof
\end{center}
\begin{center}
    \AxiomC{\(\trans{\Gamma}{e_1 : \rho}{t_1}\)}
    \AxiomC{\(\trans{\Gamma}{e_2 : \rho}{t_2}\)}
\RightLabel{\smaller (\textsc{C-If})}
\BinaryInfC{\(\trans{\Gamma}{\ifz x \then e_1 \els e_2 : \rho}{\tifz x \tthen t_1 \tels t_2}\)}
\DisplayProof
\end{center}
\begin{center}
\AxiomC{}
\RightLabel{\smaller (\textsc{C-Fail})}
\UnaryInfC{\(\trans{\Gamma}{\fail : \tau}{\tfail}\)}
\DisplayProof
\end{center}
\begin{center}
    \AxiomC{\(\trans{\Gamma, x:\tref{\alpha}{1}{\emptyset}, y:\tint}{e : \rho}{t}\)}
\RightLabel{\smaller (\textsc{C-MkRef})}
\UnaryInfC{\(\trans{\Gamma, y:\tint}{\letin{x}{\mkref y} e : \rho}{\tletin{x}{\tuple{y, \_}} t}\)}
\DisplayProof
\end{center}
\begin{center}
    \AxiomC{\(\trans{\Gamma, x:\tint, y:\tref{\alpha}{r}{B}}{e : \rho}{t}\)}
    \AxiomC{\(r > 0\)}
\RightLabel{\smaller (\textsc{C-Deref-Pos})}
\BinaryInfC{\(\trans{\Gamma, y:\tref{\alpha}{r}{B}}{\letin{x}{\star y} e : \rho}{\tletin{x}{\fst y}} t\)}
\DisplayProof
\end{center}
\begin{center}
    \AxiomC{\(\trans{\Gamma, x:\tint, y:\tref{\alpha}{0}{B}}{e : \rho}{t}\)}
\RightLabel{\smaller (\textsc{C-Deref-Zero})}
\UnaryInfC{\(\trans{\Gamma, y:\tref{\alpha}{0}{B}}{\letin{x}{\star y} e : \rho}{\tletin{x}{\_}} t\)}
\DisplayProof
\end{center}
\begin{center}
  \AxiomC{\(\trans{\Gamma}{e : \rho}{t}\)}
  \AxiomC{\(\own_{\Gamma}(x) = 1\)}
\RightLabel{\smaller (\textsc{C-Assign})}
\BinaryInfC{\(\trans{\Gamma}{x := y; e : \rho}{\tletin{x}{\tuple{y, \snd x}}} t\)}
\DisplayProof
\end{center}
\begin{center}
  \AxiomC{\(\tau_x + \tau_y = \tref \alpha r B \quad \tref \alpha r B = \rho_x + \rho_y\)}
  \noLine
  \UnaryInfC{\(\trans{\Gamma, x:\rho_x, y:\rho_y}{e : \rho}{t} \quad t' = \convalias{\tau_x \shortrightarrow \rho_x}{\tau_x \shortrightarrow \rho_y}(x, y, t)\)}
\RightLabel{\smaller (\textsc{C-Alias})}
\UnaryInfC{\(\trans{\Gamma, x:\tau_x, y:\tau_y}{\alias{x = y}; e :\rho}{t'}\)}
\DisplayProof
\end{center}
\begin{center}
  \AxiomC{\(\mathcal{L}' \mid \Gamma \vdash {e : \rho} \Rightarrow {t} \)}
  \AxiomC{\(\mathcal L' = \mathcal L \cup \{ \alpha \}\)}
\RightLabel{\smaller (\textsc{C-NewLft})}
\BinaryInfC{\(\trans{\Gamma}{\newlftin{\alpha} e}{t}\)}
\DisplayProof
\end{center}
\begin{center}
    \AxiomC{\(\ctxterm{\mathcal{L}}{\alpha} \mid \ctxterm{\Gamma}{\alpha} \vdash {e : \rho} \Rightarrow {t}\ \)}
    \AxiomC{\( \alpha = \min(\mathcal L)\)}
    \AxiomC{\(\Gamma\backslash\ctxterm{\Gamma}{\alpha} = \{x_1, \dots, x_n\}\)}
\RightLabel{\smaller (\textsc{C-EndLft})}
\TrinaryInfC{\(\trans{\Gamma}{\lftend \alpha; e : \rho}{\mathbf{assume}_{1\le i \le n}(\fst x_i = \snd x_i); t}\)}
\DisplayProof
\end{center}
\caption{Translation rules for the intraprocedural fragment}
\label{fig:translation-intraprocedual}
\end{figure}
\cref{fig:translation-intraprocedual} shows the translation rules for the intraprocedual fragment of the source language.
As in \cref{section:translation-rules}, we omit the environments \( \fenv \), \( \mathcal{L}' \) and \( \tyenv' \) to simplify the notation.
(Hence, the side condition \( x \notin \dom(\Gamma') \) is also omitted.)
The type \(\textbf{Nullify}(\tau)\) in \drvrule{C-Let} is a new type obtained by depriving all the ownership of the type \(\tau\).
Concretely, \(\textbf{Nullify}(\tau)\) is defined below.
\[
    \textbf{Nullify}(\tint) \defrel \tint \qquad \textbf{Nullify}(\tref{\alpha}{r}{B}) \defrel \tref{\alpha}{0}{\emptyset}
\]

\begin{figure}
\begin{center}
  \AxiomC{\( \Gamma' = y_1 : \tau_1, \ldots, y_n \)}
  \AxiomC{\(\{x_1, \dots, x_n\} = \textbf{RefVar}(\Gamma + x \COL \tau)\backslash\textbf{RefVar}(\Gamma')\)}
\RightLabel{\smaller (\textsc{C-Var})}
\BinaryInfC{\(\Theta \mid \mathcal L \mid {\Gamma + \Gamma' + x \COL \tau} \vdash {x \COL \tau} \rhd \mathcal{L} \mid \Gamma' \Rightarrow \mathbf{assume}_{1\le i \le n}(\fst x_i = \snd x_i); \tuple{x, y_1,\dots, y_n}\)}
\DisplayProof
\end{center}

\begin{center}
    \AxiomC{\(\rho = [\vv{\beta}/\vv{\alpha}]\tau, \rho_i = [\vv{\beta}/\vv{\alpha}]\tau_i, \rho_i' = [\vv{\beta}/\vv{\alpha}]\tau_i'\)}
    \noLine
    \UnaryInfC{\(\Theta(f) = \forall \vv{\alpha}:\mathcal{M}.\ \langle \tau_1, \dots, \tau_n \rangle \to \langle \tau_1', \dots, \tau_n' \mid \tau \rangle \)}

    \AxiomC{\( x \notin \dom(\Gamma')\)}
    \noLine
    \UnaryInfC{\([\vv{\beta}/\vv{\alpha}]\mathcal{M} \subseteq \mathcal{L}\)}
    \noLine

    \BinaryInfC{\(\tyjud{\Gamma, x:\rho, y_1:\rho_1', \dots, y_n:\rho_n'}{e : \xi}{\Gamma'} \Rightarrow t \)}
\noLine
\UnaryInfC{\( t' \equiv \tletin{\tuple{x, y_1, \cdots, y_n}}{f(y_1, \cdots, y_n)} t \)}
\RightLabel{\smaller (\textsc{C-Call})}
\UnaryInfC{\(\tyjud{\Gamma, y_1:\rho_1, \dots, y_n:\rho_n}{\letin{x}{f\langle\vv{\beta}\rangle(y_1, \dots, y_n)} e : \xi}{\Gamma'} \Rightarrow t' \)}
\DisplayProof
\end{center}
\begin{center}
\AxiomC{\(\Theta(f) = \forall \vv{\alpha}:\mathcal{L}.\ \langle \tau_1, \dots, \tau_n \rangle \to \langle \tau_1', \dots, \tau_n' \mid \tau \rangle\)}
    \AxiomC{\(\mathcal{L} \subseteq \vv{\alpha}\)}
    \noLine
    \BinaryInfC{\(\tyjudgement{\Theta}{\mathcal{L}}{x_1:\tau_1, \dots, x_n:\tau_n}{e : \tau}{\mathcal{L}}{x_1:\tau_1', \dots, x_n:\tau_n'} \Rightarrow t\)}
\RightLabel{\smaller (\textsc{C-FunDef})}
\UnaryInfC{\(\Theta \vdash f \mapsto \tuple{\vv{\alpha}}(x_1, \dots, x_n) e \Rightarrow f \mapsto (x_1, \ldots, x_n) t\)}
\DisplayProof
\end{center}
\begin{center}
    \AxiomC{\(\trans{\Gamma'}{e \COL \tau}{t}\)}
    \AxiomC{\(t' \equiv \tletin{\tuple{x, y_1, \cdots, y_n}}{f(y_1, \cdots, y_n)} t\)}
\end{center}
\caption{Translation rules for the interprocedural part}
\label{fig:translation-function}
\end{figure}

\cref{fig:translation-function} shows the rules related to function calls and returns.
We include the translation for variables in this figure as a variable is the program point where the function returns.
For these rules, we do not omit the environments \( \Theta \), \( \mathcal L' \), and \( \Gamma' \) because they play roles in the translation.
The tuple \(\tuple{x_1, x_2, x_3, \dots}\) in \drvrule{C-Var} is an abbreviation of a nested pair \(\tuple{x_1, \tuple{x_2, \tuple{x_3, \dots}}}\).
After the translation, a function returns its arguments in addition to the original returned value because the arguments may be updated in the function body.
The \(\textbf{RefVar}(\Gamma)\) indicates all variable names that are typed as references in the typing context \(\Gamma\).
We add  \( \mathbf{assume}_{1\le i \le n}(\fst x_i = \snd x_i) \) for the references that are dropped at this point.

\section{Proof of Soundness (Theorem~\ref{thm:translation-soundness})}
\label{sec:soundness-proof}
This section proves the soundness of our type-directed translation.
As in the body of the paper, we do not consider the interprocedural rules of the reduction relation for simplicity.

The following is a technical definition that corresponds to the idea of PFZ.
Since our type system allows reborrowing, we define a relation over \emph{sequence} of (aliasing) references rather than considering a relation over two references.
\begin{definition}[Prophecy Chain] \label{defn:prophecy-chain}
    Let \(\Gamma\) be a typing context, \(R\) be a register in the source language, and \( S \) be a register in the target language.
    A sequence \(x_1:\tau_1, \dots, x_n:\tau_n \in \Gamma\) is called a \emph{prophecy chain} and denoted \(\mathbf{Chain}_{\Gamma}^{R, S}(x_1, \dots, x_n)\) iff
    \begin{enumerate}
        \item \(R(x_1) = \dots = R(x_n) \in \addr\)
        \item \(\own(\tau_n) > 0 \wedge \own(\tau_1) = \dots = \own(\tau_{n - 1}) = 0\)
        \item \(\pi_1(S(x_i)) = \pi_2(S(x_{i + 1})) \quad (\text{for all } i = 1, \dots, n - 1)\)
        \item \( \lft_\Gamma(x_{i + 1}) = \blft_{\Gamma}(x_i)\) when \( \blft_{\Gamma}(x_i) \) is defined, and otherwise, \( \lft_\Gamma(x_{i + 1}) = \lft_\Gamma(x_i)\) \quad \(( \text{for all } i=1, \ldots, n - 1)\)
    \end{enumerate}
    We write \(\mathbf{Chain}_{\Gamma}^{R,\toml{S}}(x_1, \dots, x_n) \) if \(\mathbf{Chain}_{\Gamma}^{R, S}(x_1, \dots, x_n)\) holds for each \( S \in \toml S \).
    \qed
  \end{definition}

We now formally define the simulation:
\begin{definition} \label{defn:sim-relation}
    Let \(\tuple{H, R, e}\) and \(\tuple{\toml{S}, t}\) be any configurations of source and target language, respectively.
    Configurations \(\tuple{H, R, e}\) and \(\tuple{\toml{S}, t}\) are called to be in \textit{simulating relation} and denoted \(\tuple{H, R, e} \simrel_{\mathcal{L}, \Gamma} \tuple{\toml{S}, t}\) iff
    \begin{enumerate}
        \item (\textsc{S-Conv}) \(\trans {\Gamma}{e \COL \tau}{t}\)
        \item (\textsc{S-Prop}) \(\forall x:\tref{\alpha}{0}{\beta, r} \in \Gamma.\ \exists \{y_i\}_{i = 1}^n \subseteq \dom(\Gamma)\text{ such that }\) \\
        \(\mathbf{Chain}_{\Gamma}^{R, \toml{S}}(x, y_1, \dots, y_n)\)
    \end{enumerate}
        and for any \( (n_1, \ldots, n_d) \in \mathbb{Z}^d\) there exists \(S_{n_1, \ldots, n_d} \in \mathcal{S}\) such that
    \begin{enumerate}
        \setcounter{enumi}{3}
        \item (\textsc{S-Rich}) \( \pi_2(S_{n_1, \ldots, n_d}(x_i)) = n_i \quad(i = 1, \dots, d)\)
        \item (\textsc{S-Int}) \(\forall x:\tint \in \Gamma.\ S_{n_1, \ldots, n_d}(x) = R(x)\)
        \item (\textsc{S-Acc}) \(\forall x:\tref{\alpha}{r}{B} \in \Gamma .\ r > 0 \Rightarrow \pi_1(S_{n_1, \ldots, n_d}(x)) = H(R(x))\)
    \end{enumerate}
    provided that \(\{x_1, \dots, x_d\} = \{x \in \dom(\Gamma) \mid x : \tref{\alpha}{r}{B} \in \Gamma\}\).
\end{definition}
    We often omit the indices \(\mathcal{L}\) and/or \(\Gamma\) in \(\simrel_{\mathcal{L}, \Gamma}\) when they are not important or the environments we are referring to are clear from the context.
    If \(\tuple{H, R, T} \simrel_{\mathcal{L}, \Gamma} \tuple{\toml{S}, t}\), then we have \(\toml{S} \neq \emptyset\) by the existence of \( S_{\seq n}\); note that we require such register to exist even when \(d = 0\).
    The condition \drvrule{S-Rich} says that any integer can be a future value of a reference.
    In other words, it assures that we are not accidentally cutting off some non-deterministic branches.
    The condition \drvrule{S-Acc} corresponds to the TXZ principle we explained in \cref{section:translation-rules}.

 \begin{remark}[On the simplification]
   The above definition is tailored for the intraprocedural fragment.
   If we consider function calls/returns, then we need to (i) take stacks into account and (ii) keep the ``discarded environment \( \Delta \)'' as we did in \refappendix{sec:type-proof}.
 \end{remark}

\begin{lemma}[Simulation] \label{lem:simulation}
    Assume
    \[
        \tuple{H, R, e} \simrel_{\mathcal{L}, \Gamma} \tuple{\toml{S}, t} \text{ and } \tuple{H, R, e} \redsrc \tuple{H', R', e'}
    \]
    Then there exists a configuration of the target language \(\tuple{\toml{S}', t'}\), a lifetime environment \(\mathcal{L}\), and a type environment \(\Gamma'\) such that
    \[
        \tuple{H', R', e'} \simrel_{\mathcal{L}', \Gamma'} \tuple{\toml{S}', t'} \text{ and } \tuple{\toml{S}, t} \redmultgt \tuple{\toml{S}', t'}
    \]
\end{lemma}
\begin{proof}
  \newcommand*{\Sc}{S_{\seq n}}
  We prove this by case analysis on the reduction \(\tuple{H, R, e} \redsrc \tuple{H', R', e'}\).
  In each case, once we choose \( t' \), \( {\mathcal L}' \) and \( \Gamma' \), we need to show that \(\trans[\mathcal{L}']{\Gamma'}{e' \COL \tau}{t'}\).
  However, we will omit this argument as it is essentially identical to the argument we did in the proof of \cref{thm:ownsum-bound}.
  (We also implicitly assume that \( \cons \lenv \tyenv H R \) holds.)
  We focus on the proof of the preservation for conditions such as \drvrule{S-Prop} or \drvrule{S-Acc}, as these are the key properties of the translation.
  We start with the important cases.
  \subsubsection*{Case: \sfdrvrule{Rs-Alias}}
    In this case, the translation is given by the following derivation.
    \begin{center}
        \AxiomC{\(\trans{\Gamma_0, x:\rho_x, y:\rho_y}{e : \rho}{t'}\)}
        \AxiomC{\(t = \convalias{\tau_x \shortrightarrow \rho_x}{\tau_y \shortrightarrow \rho_y}(x, y, t')\)}
    \RightLabel{\smaller (\textsc{C-Alias})}
    \BinaryInfC{\(\trans{\Gamma_0, x:\tau_x, y:\tau_y}{\alias{x = y}; e :\rho}{t}\)}
    \DisplayProof
    \end{center}
    Note that \(\own(\tau_x) + \own(\tau_y) = \own(\rho_x) + \own(\rho_y)\).
    As expected, we take \( \Gamma_0, x:\rho_x, y:\rho_y \) as \( \Gamma' \) and \( \mathcal L  \) as \( \mathcal {L}' \).
    We do a case analysis according to the cases used in the definition of \( \convalias{\tau_x \shortrightarrow \rho_x}{\tau_y \shortrightarrow \rho_y} \).

    \begin{description}
    \item[Subcase: \( \own(\tau_x) = \own(\rho_x) = 0\)]~\\
        This case is trivial.
    \item[Subcase: \(\own(\rho_x), \own(\rho_y) > 0\) and \( \own(\tau_y) > 0\)]~\\
    In this case, we choose the following matching reduction:
    \[
        \tuple{\mathcal{S}, t} \equiv \tuple{\mathcal{S}, \tletin{x}{\tuple{\fst y, \snd x}} t'} \redmultgt \tuple{\mathcal{S}', t'}
    \]
    where \(\toml{S}' := \{ S\{x \mapsto \tuple{\pi_1(S(y)), \pi_2(S(x))}\} \mid S \in \toml{S} \}\).
    We have \(\pi_1(\Sc(y)) = H(R(y))\) for the registers \( \Sc \in \toml S \) satisfying \drvrule{S-Acc} (as well as \drvrule{S-Rich} and \drvrule{S-Int}), which implies that \( \Sc\{x \mapsto \tuple{\pi_1(\Sc(y)), \pi_2(\Sc(x))}\} \) also satisfy \drvrule{S-Acc}.
    It is easy to see that \( \Sc\{x \mapsto \tuple{\pi_1(\Sc(y)), \pi_2(\Sc(x))}\} \) satisfies \drvrule{S-Rich} and \drvrule{S-Int} as well.
    We have to show that \drvrule{S-Prop} is preserved.
    It suffices to consider all chains containing \(x\).
    Let \(y_1, \dots,y_n \in \dom(\Gamma)\) such that \(\mathbf{Chain}_{\Gamma}^{R, \toml{S}}(y_1, \dots, y_n)\) and \(x = y_k\) for some \(k = 1, \dots, n\).
    Then \(\mathbf{Chain}_{\Gamma'}^{R, \toml{S}'}(y_1, \dots, y_k)\) holds so we can use this chain instead.

    \item[Subcase: \(\own(\tau_x) > 0\) and \(\own(\rho_x) = 0\)]~\\
    We have \(\own(\rho_y) > 0\) and
    \[
        t \equiv \left\{\begin{array}{l}
            \tletin{y}{\tuple{\fst x, \snd y}} \\
            \tletin{x}{\tuple{\snd y, \snd x}} \\
            t'
        \end{array}\right.
    \]
    for some \( t' \).
    Therefore we take \(\tuple{\mathcal{S}, t} \redmultgt \tuple{\mathcal{S}', t'}\), as the matching transition, where
    \[
        \toml{S}' := \big\{ S\{x \mapsto \tuple{\pi_2(S(y)), \pi_2(S(x))},\, y \mapsto \tuple{\pi_1(S(x)), \pi_2(S(y))}\} \mid S \in \toml{S} \big\}
    \]
    For registers \( \Sc \) that satisfy the condition such as \drvrule{S-Acc}, let us define
    \[
      \Sc' \defrel \Sc\{x \mapsto \tuple{\pi_2(\Sc(y)), \pi_2(\Sc(x))},\, y \mapsto \tuple{\pi_1(\Sc(x)) \pi_2(\Sc(y)) }.
    \]
    We have
    \begin{align*}
      \pi_1(\Sc'(y))
      &= \pi_1(\Sc(x)) \tag{by the def.\ of \( \Sc'(y)\)} \\
      &= H(R(x)) \tag{by (\textsc{S-Acc})} \\
      &= H(R(y)) \tag{since \( y \) is an alias of \( x \)}.
    \end{align*}
    Hence, \drvrule{S-Acc} also holds for every \(\Sc'\), and we can choose these registers as the registers that satisfy the conditions \drvrule{S-Rich}, \drvrule{S-Int} and \drvrule{S-Acc}.
    To show that \drvrule{S-Prop} is preserved, we have to confirm that for any \(z \in \mathbf{RefVar}(\Gamma')\) such that \(\own_{\Gamma'}(z) = 0\), there exists a prophecy chain \(z, w_1, \dots, w_m \in \dom(\Gamma)\) on \(\toml{S}'\).
    If \(z = x\), \(\textbf{Chain}_{\Gamma'}^{R, \toml{S}'}(z, y)\) holds.
    Otherwise, \(\own_{\Gamma}(z) = 0\), which implies that by \drvrule{S-Prop}, there exists \(z_1, \dots, z_n \in \dom(\Gamma')\) such that \(\textbf{Chain}_{\Gamma}^{R, \toml{S}}(z, z_1, \dots, z_n)\).
    We assume \(x \in \{z, z_1, \dots, z_n\}\) because otherwise, it is obvious as we can take this prophecy chain also on \( \toml{S}'\).
    Then we must have \(z_n = x\) since \(\own(\tau_x) > 0\). %
    In addition, \(\pi_1(S'(x)) = \pi_2(S'(y))\) holds for any \(S' \in \toml{S}'\).
    If \( y \) borrowed the ownership from \( x \), then we have \( \blft_{\Gamma'}(x) = \lft_{\Gamma'}(y)\); otherwise the ownership was shared and we have \( \lft_{\Gamma'}(x) = \lft_{\Gamma'}(y)\).
    Therefore, we have \(\textbf{Chain}_{\Gamma'}^{R, \toml{S}'}(z, z_1, \dots, z_{n-1}, x, y)\), which is a prophecy chain for \( z \) on \( \toml{S}'\).

      \item[Remaining subcases:]~\\
        Symmetric to the previous cases.
    \end{description}

    \subsubsection*{Case: \sfdrvrule{Rs-Endlft}}
    In this case, we must have
    \[
      t \equiv \mathbf{assume}_{1\le i \le n}(\fst x_i = \snd x_i); t',
    \]
    for some \( t' \) where \(\{x_1, \dots, x_n\} = \Gamma\backslash\ctxterm{\Gamma}{\alpha}\).
    As the matching transition, we take \(\tuple{\toml{S}, t} \redmultgt[(\textsc{Rt-Assume})] \tuple{\toml{S}', t'}\)  where
    \begin{align*}
        \toml{S}' := \big\{S \in \toml{S} \mid \pi_1(S(x_i)) = \pi_2(S(x_i))\ (i = 1, \dots, n)\big\}
    \end{align*}
    Note that we can chose  \(\mathcal{L}' := \ctxterm{\mathcal L}{\alpha}\) and \(\Gamma' := \ctxterm{\Gamma}{\alpha}\) in this case.

    We show that \drvrule{S-Acc} and \drvrule{S-Prop} hold for \( \tuple{\toml{S}', t'} \); the other conditions are easy to check.

    First, we show that \drvrule{S-Prop} holds.
    Let \(x : \tau' \in \Gamma'\) such that \( \tau' \) is a reference type.
    We need to find a prophecy chain starting from \( x \).
    By the definition of \( \Gamma' \), \( \tau' = \ctxterm{\tau}{\alpha}\) for some \(\tau = \tref{\beta}{r}{B}\) such that \(\beta \neq \alpha\).
    When \(r > 0\), it is straightforward.
    We consider the case where \(\tau = \tref{\beta}{r}{B}\) (\(\beta \neq \alpha\)) and \(r = 0\).
    We start by taking a prophecy chain of \( \toml{S}\); \drvrule{S-Prop} of \( \tuple{H, R, e} \simrel_{\mathcal{L}, \Gamma} \tuple{\toml{S}, t} \) implies that there exists \(y_1, \dots, y_n \in \dom(\Gamma)\) such that \(\mathbf{Chain}_{\Gamma}^{R, \toml{S}}(x, y_1, \dots, y_n)\).
    Let us consider the subsequence of \(y_1, \dots, y_n\) obtained by removing all elements such that \(\lft_{\Gamma}(y_i) = \alpha\).
    By the minimality of \( \alpha \) such a subsequence is a prefix of \( y_1, \ldots, y_n \), say \( y_1, \ldots, y_m \) and \( \blft_\Gamma(y_m) = \alpha \).
    Hence we have \( \own_{\Gamma'}(y_m) > 0 \), which implies that \(  \mathbf{Chain}_{\Gamma'}^{R, \toml{S'}}(x, y_1, \dots, y_m) \).

   To show \drvrule{S-Acc} holds, it suffices to show that \(\pi_1(\Sc(x)) = H(R(x))\) for \(x : \tref{\beta}{0}{\alpha, s} \in \Gamma\) (\(s > 0\)), which is a reference that regained non-zero ownership after the lifetime termination.
    Here \( \Sc \in \toml{S} \) is the register that satisfies \drvrule{S-Rich}, \drvrule{S-Int} and \drvrule{S-Acc}, and moreover, we assume that \( \Sc \in \toml{S}' \).
    (In other words, \( \Sc \) is the register that correctly guessed the future values of \( x_1, \ldots, x_n \), which always exists.)
    Observe that there is a prophecy chain \( \textbf{Chain}_{\Gamma}^{R, \toml{S}}(x, y_1, \dots, y_n) \) such that \( \lft_\Gamma(y_1) = \cdots = \lft_\Gamma(y_n) = \alpha \) because \( \alpha \) is a minimal lifetime.
    In this case, we have
    \begin{align*}
      \pi_1(\Sc(x))
      &= \pi_2(\Sc(y_1)) \qquad (\text{by } \textbf{Chain}_{\Gamma}^{R, \toml{S}}(x, y_1, \dots, y_n)) \\
      &= \pi_1(\Sc(y_1)) \qquad (\text{since } \Sc \in \toml{S}') \\
      &\qquad \vdots \\
      &= \pi_1(\Sc(y_{n - 1})) \\
      &= \pi_2(\Sc(y_n)) \\
      &= \pi_1(\Sc(y_n)) \tag{\( \Sc \in \toml{S}' \)}\\
      &= H(R(y_n)) \tag{by \drvrule{S-Acc}} \\
      &= H(R(x)) \tag{since \( R(x) = R(y_n) \)}
    \end{align*}
    Therefore, \drvrule{S-Acc} also stands.
    Hence \(\tuple{H',R',T'} \simrel_{\mathcal{L}', \Gamma'} \tuple{S', t'}\) holds.

    \subsubsection*{Case: \sfdrvrule{Rs-Assign}}
    We only show the preservation of \drvrule{S-Acc} because preservations of the other conditions are straightforward to see.
    In this case, there exists \(x, y \in \setvar\) such that \(t \equiv \letin{x}{\tuple{y, \snd x}} t'\).
    We take the matching reduction from \(\tuple{\toml{S}, t}\) as follows
    \[
        \tuple{\toml{S}, t} \redmultgt \tuple{\{S\{x \hookrightarrow \tuple{S(y), \pi_2(S(x))}\} \mid S \in \toml{S}\}, t'}
    \]
    For each \( \seq n \in {\mathbb Z} ^d \), the register \( \Sc \in \toml S\) that satisfy \drvrule{S-Int} is updated as \( \Sc' \defrel \Sc\{x \hookrightarrow \tuple{R(y), \pi_2(\Sc(x))}\} \) because \( \Sc(y) = R(y) \).
    Hence, we have \( \pi_1( \Sc'(x) ) = R(y) = H'(R(y)) \) as desired.
    It remains to consider references aliasing to \( x \).
    By \cref{thm:ownsum-bound}, no references are aliasing to \(x\) with positive ownership.
    Thus (\textsc{S-Acc}) holds for these aliasing registers.

    \subsubsection*{Case: \sfdrvrule{Rs-Deref}}
    In this case, \( e \) must be of the form \( \letin{x}{\star y} e' \).
    We must also have \( H' = H \) and \( R' = R \{ x \mapsto H(R(y)) \}\), where \( H(R(y))  \in \mathbb Z \).
    We proceed by a case analysis on \( \own_\Gamma(y) \).

    If \( \own_\Gamma(x) > 0\), then the last rule applied for the derivation of the translation must be (\textsc{C-Deref-Pos}).
    So we must have \( t \equiv \tletin{x}{\fst y} t'\) for some \( t' \).
    We take the matching transition as
    \[
      \tuple{\toml{S},  \tletin{x}{\fst y} t'} \redmultgt \tuple{\{S\{x \mapsto \pi_1(S(y))\} \mid S \in \toml{S} \}, t' }
    \]
    It is easy to see that \( \tuple{ \toml{S}', t'} \) satisfies (\textsc{S-Prop}), and we are left to choose \( \Sc' \in \toml{S}' \) satisfying (\textsc{S-Rich}), (\textsc{S-Int}) and (\textsc{S-Acc}).
    For any \( (n_1, \ldots, n_d) \in {\mathbb Z}^d \), we take \( \Sc \{ x \mapsto H(R(y)) \} \) as such register, where \( \Sc \in S \) is the register that satisfies (\textsc{S-Rich}), (\textsc{S-Int}) and (\textsc{S-Acc}).
    Note that \( \Sc \{ x \mapsto H(R(y)) \} \in \toml{S}' \) because \( \Sc \in \toml S \) and \( \pi_1 (\Sc(y)) = H(R(y)) \) by (\textsc{S-Acc}).
    The register  \( \Sc \{ x \mapsto H(R(y)) \} \) satisfies (\textsc{S-Rich}) and (\textsc{S-Acc}) because \( \Sc \) does.
    Obviously, (\textsc{S-Acc}) also holds because \( \ H(R(y)) = R'(x) \).

    Now we consider the case where \( \own_\Gamma(y) =  0\).
    The last rule applied for the derivation of the translation must be (\textsc{C-Deref-Zero}), and thus \( t \) must be of the shape \(  \tletin{x}{\_} t'\).
    In this case, we use (\textsc{Rt-LetNondet}) to match the transition.
    That is, we take
    \[
      \tuple{\toml{S},  \tletin{x}{\_} t'} \redtgt \tuple{\{S\{x \mapsto n\} \mid S \in \toml{S}, n \in \mathbb{Z}\}, t' }
    \]
    as the matching transition.
    As in the case where \( \own_\Gamma(y) > 0 \), the only non-triviality is how to choose \( \Sc' \in \toml{S}' \) satisfying (\textsc{S-Rich}), (\textsc{S-Int}) and (\textsc{S-Acc}) for each \( (n_1, \ldots, n_d) \in {\mathbb Z}^d \).
    Again we take \( \Sc \{ x \mapsto H(R(y)) \} \) as such register, where \( \Sc \in S \) is the register that satisfies (\textsc{S-Rich}), (\textsc{S-Int}) and (\textsc{S-Acc}).
    Note that \( \Sc \{ x \mapsto H(R(y)) \} \in \toml{S'} \) because \( \toml{S}'\) contains all the possible integer assignment to \( x \).

    \subsubsection*{Cases: \sfdrvrule{Rs-IfTrue}, \sfdrvrule{Rs-IfFalse}}
    We only give proof for \drvrule{Rs-IfTrue}.
    In this case, \(e \equiv \ifz x \then e_1 \els e_2\), \(t \equiv \ifz x \then t_1 \els t_2\), and \(e' = e_1\).
    We take
    \[
        \tuple{\mathcal{S}, t} = \tuple{\mathcal{S}, \tifz x \tthen t_1 \tels t_2} \redtgt[\drvrule{Rt-IfTrue}] \tuple{\mathcal{S'}, t_1}
    \]
    where
    \[
      \toml{S}' = \{ S \in \toml{S} \mid S(x) = 0\}
    \]
    as the matching transition.
    Note that we have \( \Sc \in \mathcal{S}' \) for each \( \seq n \) because \( \Sc(x) = R(x) = 0 \).
    Hence, \(\tuple{H, R, e_1} \simrel_{\mathcal{L}, \Gamma} \tuple{\toml{S}', t_1}\) is straightforward to check.

    \subsubsection*{Case: \sfdrvrule{Rs-Let}}
    This case is easily reduced to the \drvrule{Rs-Alias} case.

    \subsubsection*{Case: \sfdrvrule{Rs-Arith}, \sfdrvrule{Rs-Newlft}, \sfdrvrule{Rs-MkRef}}
    These cases are trivial.

    \subsubsection*{Cases: \sfdrvrule{Rs-Var} \sfdrvrule{Rs-Call}, \sfdrvrule{Rs-AliasFail}}
    These cases do not happen because we omit functions and assume progression.

   \qed
\end{proof}

\newpage
\section{Benchmark programs} \label{sec:benchmark-programs}

\begin{figure}
\begin{tabular}{ll}
\begin{minipage}[t]{0.42\hsize}
\begin{lstlisting}[style=mystyle]
rand_choose(x, y) {
  if _ then x else y
}

let a = _ in
if a < 0 then () else
let x = mkref a in
let y = rand_choose(x, mkref _) in
let z = y in
let b = *y + *z in
alias(y = z);
y := *y + b;
assert( *x <= 3 * a )
\end{lstlisting}
    \caption{original `borrow-merge'}
\end{minipage} &
\begin{minipage}[t]{0.54\hsize}
\begin{lstlisting}[style=mystyle]
let nondet () = Random.int(0)
let rec assume x n =
  if x = n then () else assume x n

let rand_choose x y =
  if nondet() < 0 then (
    let x' = (fst x, nondet()) in
    let x = (snd x', snd x) in
    (x', x, y)
  ) else (
    let y' = (fst y, nondet()) in
    let y = (snd y', snd y) in
    (y', x, y)
  )

let main a = (
  if a < 0 then () else
  let x = (a, nondet()) in
  let w = (nondet(), nondet()) in
  let (y, x, w) = rand_choose x w in
  let z = (fst y, nondet()) in
  let b = fst y + fst z in
  let y = (fst z, snd y) in
  let z = (snd y, snd z) in
  let y = (fst y + b, snd y) in
  assume (fst y) (snd y);
  assume (fst z) (snd z);
  assume (fst w) (snd w);
  assert(fst x <= 3 * a)
)
\end{lstlisting}
    \caption{converted `borrow-merge'}
\end{minipage}
\end{tabular}
\end{figure}

\begin{figure}
\begin{tabular}{ll}
\begin{minipage}[t]{0.42\hsize}
\begin{lstlisting}[style=mystyle]
loop(a, b) {
  let a_ = *a in
  let b_ = *b in
  b := *b + 1;
  a := *a + 1;
  assert( *a = (a_ + 1));
  if _ then
    loop(b, mkref _)
  else
    loop(b, a)
}
loop((mkref _), (mkref _))
\end{lstlisting}
    \caption{original `simple-loop' (safe)}
\end{minipage} &
\begin{minipage}[t]{0.54\hsize}
\begin{lstlisting}[style=mystyle]
let nondet () = Random.int(0)
let rec assume x n =
  if x = n then () else assume x n

let rec loop a b = 
  let a_ = fst a in
  let b = (fst b + 1, snd b) in
  let a = (fst a + 1, snd a) in
  assert(fst a = a_ + 1);
  if nondet() < 0 then
    loop b (nondet(), nondet())
  else
    loop b a

let main =
  loop (nondet(), nondet()) (nondet(), nondet())
\end{lstlisting}
    \caption{converted `simple-loop' (safe)}
\end{minipage}
\end{tabular}
\end{figure}

\begin{figure}
\begin{tabular}{ll}
\begin{minipage}[t]{0.42\hsize}
\begin{lstlisting}[style=mystyle]
just_rec(x) {
  if _ then () else (
    let y = mkref _ in
    just_rec(x)
  )
}
let x = mkref _ in
let x0 = *x in
just_rec(x);
assert(x0 = *x)
\end{lstlisting}
    \caption{original `just-rec' (safe)}
\end{minipage} &
\begin{minipage}[t]{0.54\hsize}
\begin{lstlisting}[style=mystyle]
let nondet () = Random.int(0)
let rec assume x n =
  if x = n then () else assume x n

let rec just_rec x =
  if nondet() < 0 then ((), x) else
    let y = (nondet(), nondet()) in
    let (r, y) = just_rec y in
    (r, x)

let main =
  let x = (nondet(), nondet()) in
  let x0 = fst x in
  let (_, x) = just_rec x in
  assert(fst x = x0)
\end{lstlisting}
    \caption{converted `just-rec' (safe)}
\end{minipage}
\end{tabular}
\end{figure}

\begin{figure}
\begin{tabular}{ll}
\begin{minipage}[t]{0.42\hsize}
\begin{lstlisting}[style=mystyle]
f(p, q) {
  alias(p = q);
  q := *q + 3
}

let x = mkref 1 in 
let y = x in
assert( *y = *x );
f(x, y);
assert(*x = 4)
\end{lstlisting}
    \caption{original `shuffle-in-call' (safe)}
\end{minipage} &
\begin{minipage}[t]{0.54\hsize}
\begin{lstlisting}[style=mystyle]
let nondet () = Random.int(0)
let rec assume x n =
  if x = n then () else assume x n

let rec f p q =
  let q = (fst p, snd q) in
  let p = (snd q, snd p) in
  let q = (fst q + 3, snd q) in
  ((), p, q)

let main =
  let x = (1, nondet()) in
  let y = (fst x, nondet()) in
  assert(fst x = fst y);
  let (_, x, y) = f x y in
  assume (fst y) (snd y);
  assert(fst x = 4)
\end{lstlisting}
    \caption{converted `shuffle-in-call' (safe)}
\end{minipage}
\end{tabular}
\end{figure}

\begin{figure}
\begin{tabular}{ll}
\begin{minipage}[t]{0.42\hsize}
\begin{lstlisting}[style=mystyle]
takemax(x, y) {
  if *x >= *y then x else y
}
let x = mkref _ in
let y = mkref _ in
let z = takemax(x, y) in
z := *z + 1;
assert( *x != *y )
\end{lstlisting}
    \caption{original `inc-max' (safe)}
\end{minipage} &
\begin{minipage}[t]{0.54\hsize}
\begin{lstlisting}[style=mystyle]
let nondet () = Random.int(0)
let rec assume x n =
  if x = n then () else assume x n

let takemax x y =
  if fst x >= fst y then (
    let x' = (fst x, nondet()) in
    let x = (snd x', snd x) in
    (x', x, y)
  ) else (
    let y' = (fst y, nondet()) in
    let y = (snd y', snd y) in
    (y', x, y)
  )

let main =
  let x = (nondet(), nondet()) in
  let y = (nondet(), nondet()) in
  let (z, x, y) = takemax x y in
  let z = ((fst z) + 1, snd z) in
  assume (fst z) (snd z);
  assert (fst x <> fst y)
\end{lstlisting}
    \caption{converted `inc-max' (safe)}
\end{minipage}
\end{tabular}
\end{figure}

\begin{figure}
\begin{tabular}{ll}
\begin{minipage}[t]{0.42\hsize}
\begin{lstlisting}[style=mystyle]
f(p, q) {
  alias(p = q);
  q := *q + 3
}

let x = mkref 1 in 
let y = x in
assert( *y = *x );
f(x, y);
assert( *x = 4 )
\end{lstlisting}
    \caption{original `shuffle-in-call' (safe)}
\end{minipage} &
\begin{minipage}[t]{0.54\hsize}
\begin{lstlisting}[style=mystyle]
let nondet () = Random.int(0)
let rec assume x n =
  if x = n then () else assume x n

let rec f p q =
  let q = (fst p, snd q) in
  let p = (snd q, snd p) in
  let q = (fst q + 3, snd q) in
  ((), p, q)

let main =
  let x = (1, nondet()) in
  let y = (fst x, nondet()) in
  assert(fst x = fst y);
  let (_, x, y) = f x y in
  assume (fst y) (snd y);
  assert(fst x = 4)
\end{lstlisting}
    \caption{converted `shuffle-in-call' (safe)}
\end{minipage}
\end{tabular}
\end{figure}

\begin{figure}
\begin{tabular}{ll}
\begin{minipage}[t]{0.36\hsize}
\begin{lstlisting}[style=mystyle]
loop(res, cnt) {
  if *cnt > 0 (
    cnt := cnt - 1;
    res := res + 1;
    loop(res, cnt)
  ) else ()
}

let a = _ in
let b = _ in
let res = mkref a in
let cnt = mkref b in
if b < 0 then () else (
  loop(res, cnt);
  assert(a + b = *res)
)
\end{lstlisting}
    \caption{original `hhk2008'}
\end{minipage} &
\begin{minipage}[t]{0.60\hsize}
\begin{lstlisting}[style=mystyle]
let nondet () = Random.int(0)
let rec assume x n =
  if x = n then () else assume x n

let rec loop res cnt =
  if fst cnt > 0 then (
    let cnt = (fst cnt - 1, snd cnt) in
    let res = (fst res + 1, snd res) in
    let (r, res, cnt) = loop res cnt in
    (r, res, cnt)
  ) else ((), res, cnt)

let main a b =
  let res = (a, nondet()) in
  let cnt = (b, nondet()) in
  if b < 0 then () else (
    let res' = (fst res, nondet()) in
    let res = (snd res', snd res) in
    let ((), res', cnt) = loop res' cnt in
    assume (fst res') (snd res');
    assert(a + b = fst res)
  )
\end{lstlisting}
    \caption{converted `hhk2008'}
\end{minipage}
\end{tabular}
\end{figure}

\begin{figure}
\begin{lstlisting}[style=mystyle]
let nondet () = Random.int(0)
let rec assume x n =
    if x = n then () else assume x n

let minmax x y =
    if fst x < fst y then (
        let x' = (fst x, nondet()) in
        let y' = (fst y, nondet()) in
        ((x', y'), x, y)
    ) else (
        let x' = (fst x, nondet()) in
        let y' = (fst y, nondet()) in
        ((y', x'), x, y)
    )

let rand_choose x y =
    if nondet() < 0 then (
        let x' = (fst x, nondet()) in
        let x = (snd x', snd x) in
        (x', x, y)
    ) else (
        let y' = (fst y, nondet()) in
        let y = (snd y', snd y) in
        (y', x, y)
    )

let main =
    let x = (nondet(), nondet()) in
    let y = (nondet(), nondet()) in
    let ((p, q), x, y) = minmax x y in
    let (z, x, y) = rand_choose x y in
    assert(fst p <= fst z && fst z <= fst q);
    assume (fst p) (snd p);
    assume (fst q) (snd q);
    let z = (1, snd z) in
    assert(fst z = 1)
\end{lstlisting}
    \caption{converted `minmax'}
\end{figure}

\begin{figure}
\begin{tabular}{ll}
\begin{minipage}[t]{0.36\hsize}
\begin{lstlisting}[style=mystyle]
rand_choose(x, y) {
  if _ then x else y
}
linger_dec(x) {
  x := *x - 1;
  if _ then () else (
    let y = mkref _ in
    linger_dec(rand_choose(x, y))
  )
}
let x = mkref _ in
let x0 = *x in
linger_dec(x);
assert(x0 > *x)
\end{lstlisting}
    \caption{original `linger-dec' (safe)}
\end{minipage} &
\begin{minipage}[t]{0.60\hsize}
\begin{lstlisting}[style=mystyle]
let nondet () = Random.int(0)
let rec assume x n =
  if x = n then () else assume x n

let rand_choose x y =
  if nondet() < 0 then (
    let x' = (fst x, nondet()) in
    let x = (snd x', snd x) in
    (x', x, y)
  ) else (
    let y' = (fst y, nondet()) in
    let y = (snd y', snd y) in
    (y', x, y)
  )

let rec linger_dec x =
  let x = (fst x - 1, snd x) in
  if nondet() < 0 then ((), x) else (
    let y = (nondet(), nondet()) in
    let (z, x, y) = rand_choose x y in
    let (r, z) = linger_dec z in
    assume (fst z) (snd z);
    (r, x)
  )

let main =
  let x = (nondet(), nondet()) in
  let x0 = fst x in
  let x' = (fst x, nondet()) in
  let x = (snd x', snd x) in
  let (_, x') = linger_dec x' in
  assume (fst x') (snd x');
  assert(x0 > fst x)
\end{lstlisting}
    \caption{converted `linger-dec' (safe)}
\end{minipage}
\end{tabular}
\end{figure}

\end{document}